\documentclass[11pt,leqno,fleqn]{article}

\newcommand{\Comment}[1]{\relax}

\usepackage{amsthm}
\usepackage{latexsym}
\usepackage{amsmath}
\usepackage{amssymb}
\usepackage{amsfonts}
\usepackage{mathtools}
\usepackage{algorithm}
\usepackage{algorithmicx}
\usepackage{amsbsy}
\usepackage{synttree} 
\usepackage{diagrams}
\usepackage{enumitem} 
\usepackage[usenames]{color}
\usepackage[colorlinks=false, %
           citecolor=red,filecolor=blue,linkcolor=blue,urlcolor=blue,  
           linktoc=page,
           pagebackref=true
           ]{hyperref} 

\usepackage{fullpage} 
\usepackage{afterpage}  
\usepackage{float} 
\usepackage{wrapfig}  
\floatstyle{boxed} 
\restylefloat{figure}   

\usepackage{times}
\usepackage{bm}   
\usepackage{stmaryrd} 
\usepackage{MnSymbol}  
\usepackage{mathrsfs}  
\usepackage{graphicx}
\usepackage[all]{xy}
\usepackage{fancybox}
\usepackage{alltt}

\newcommand{\Hide}[1]{}

\newif
\ifnote
\notetrue

\newif
\ifTR
\TRfalse

\newtheorem{theorem}{Theorem}
\newtheorem{lemma}[theorem]{Lemma}
\newtheorem{corollary}[theorem]{Corollary}
\newtheorem{proposition}[theorem]{Proposition}

\newtheorem{restxxx}[theorem]{Restriction}

\newtheorem{agreexxx}[theorem]{Agreement}

\newtheorem{termxxx}[theorem]{Terminology}

\newtheorem{notxxx}[theorem]{Notation}

\newtheorem{assumxxx}[theorem]{Assumption}

\newtheorem{convenxxx}[theorem]{Convention}

\newtheorem{exaxxx}[theorem]{Example}
\newenvironment{example}{\begin{exaxxx}\rm}{\hfill\QED\end{exaxxx}}
\newtheorem{exexxx}[theorem]{Exercise}

\newtheorem{remxxx}[theorem]{Remark}
\newenvironment{remark}{\begin{remxxx}\rm}{\hfill\QED\end{remxxx}}
\newtheorem{openxxx}[theorem]{Open Problem}
\newtheorem{conjxxx}[theorem]{Conjecture}

\newtheorem{defxxx}[theorem]{Definition}
\newenvironment{definition}[1]{\begin{defxxx}[\emph{#1}]\rm}%
{\hfill\QED\end{defxxx}}
\newtheorem{procxxx}[theorem]{Procedure}
{\hfill\QED\end{procxxx}}

\newtheorem{Prxxx}[theorem]{Proof}
{\end{Prxxx}} 

\newenvironment{custommargins}[2]%
  {\addtolength{\leftskip}{#1}\addtolength{\rightskip}{#2}}{\par}

\usepackage[margin=0pt,justification=centerlast,font=small,format=hang,%
labelfont=bf,up,textfont=rm,up]{caption}

\newcommand{\Set}[1]{\{ #1 \}}
\newcommand{\SET}[1]{\bigl\{ #1 \bigr\}}

\newcommand{\A}{{\cal A}}
\newcommand{\B}{{\cal B}}

\newcommand{\LL}{{\cal L}}

\newcommand{\sss}{{\cal S}}

\newcommand{\X}{{\cal X}}
\newcommand{\Y}{{\cal Y}}

\newcommand{\bigO}[1]{{\cal O}\bigl(#1\bigr)} 
\newcommand{\bigOO}[1]{{\cal O}(#1)} 

\newcommand{\Let}[3]%
    {\textbf{\textsf{let}}\ {#1}\,{#2}\ \textbf{\textsf{in}}\;{#3}\,}
\newcommand{\Try}[3]%
    {\textbf{\textsf{try}}\ {#1} {#2}\ \textbf{\textsf{in}}\;{#3}\;}
\newcommand{\Mix}[3]%
    {\textbf{\textsf{mix}}\ {#1} {#2}\ \textbf{\textsf{in}}\;{#3}\;}
\newcommand{\LET}[3]%
    {\textbf{\textsf{let}}^*\ {#1} {#2}\ \textbf{\textsf{in}}\;{#3}\;}
\newcommand{\Letrec}[3]%
    {\textbf{\textsf{letrec}}\ {#1} {#2}\ \textbf{\textsf{in}}\;{#3}\;}

\newcommand{\bridges}[2]{{\partial}_{#1}(#2)} 
\newcommand{\Bridges}[2]{{\partial}_{#1}\bigl(#2\bigr)} 
\newcommand{\degreeSym}{{\text{\em degree}}} 
\newcommand{\degree}[1]{{\degreeSym}(#1)}
\newcommand{\degr}[2]{{\degreeSym}_{#1}(#2)}
\newcommand{\Degr}[2]{{\degreeSym}_{#1}\bigl(#2\bigr)}
\newcommand{\heightSym}{{\text{\em height}}} 
\newcommand{\height}[2]{{\heightSym}_{#1}(#2)}

\newcommand{\scatterSym}{{\text{\em scatter}}} 
\newcommand{\scatter}[1]{{\scatterSym}(#1)} 
\newcommand{\unbalSym}{{\text{\em unbal}}} 
\newcommand{\unbal}[1]{{\unbalSym}(#1)}

\newcommand{\ie}{\textit{i.e.}}
\newcommand{\eg}{\textit{e.g.}}
\newcommand{\QED}{{\Large $\square$}} 

\newcommand{\size}[1]{|\,#1\,|}  
\newcommand{\ssize}[1]{\bigl|\,#1\,\bigr|}  
  
\newcommand{\set}[1]{\overline{#1}}

\newcommand{\SSS}{\mathscr{S}}
\newcommand{\PPP}{{\mathscr{P}}}

\newcommand{\canonicalOrd}[1]{\text{\em canon}(#1)}

\newcommand{\binary}[1]{\text{\em binary}(#1)}

\newcommand{\mincut}{\text{\rm CutWidth}} 
\newcommand{\ola}{\text{\rm MinArr}} 
\newcommand{\EdgePerm}{\Theta} 

\newcommand{\spacing}[2]{
  \renewcommand{\baselinestretch}{#2}
  \small\normalsize #1
  \setlength{\parskip}{0.1\baselineskip}
  \settowidth{\parindent}{xxxx}
  \setlength{\parindent}{#2\parindent}
  \setlength{\leftmargini}{\parindent}
  \setlength{\leftmarginii}{\parindent}
  \setlength{\leftmarginiii}{\parindent}
  \setlength{\footnotesep}{#2\footnotesep}
}

\newcommand{\circled}[1]{
   {\huge\raisebox{-1pt}{\textcircled{\raisebox{2.0pt} {\normalsize #1}}}}}

\begin{document}

\spacing{\normalsize}{0.98}
\setcounter{page}{1}     
\setcounter{tocdepth}{1} 
\ifTR
  \pagenumbering{roman} 
\else
\fi

\title{Efficient Reassembling of Graphs, Part 1: The Linear Case} 

\author{Assaf Kfoury%
           \thanks{Partially supported by NSF awards CCF-0820138
           and CNS-1135722.} \\
        Boston University \\
        \ifTR Boston, Massachusetts \\ 
        \href{mailto:kfoury@bu.edu}{kfoury{@}bu.edu}
        \else \fi
\and
       Saber Mirzaei%
          \footnotemark[1]\\
       Boston University  \\
        \ifTR Boston, Massachusetts \\ 
        \href{mailto:smirzaei@bu.edu}{smirzaei{@}bu.edu}
        \else \fi
}

\ifTR
   \date{\today}
\else
   \date{} %
\fi
\maketitle
  \ifTR
     \thispagestyle{empty} 
  \else
  \fi

\vspace{-.3in}
  \begin{abstract}

\noindent
The \emph{reassembling of a simple connected graph} $G = (V,E)$ is an
abstraction of a problem arising in earlier studies of network
analysis. Its simplest formulation is in two steps: 
\begin{itemize}[itemsep=0pt,parsep=2pt,topsep=2pt,partopsep=0pt] 
\item[(1)] We cut every edge of $G$ into two halves, 
           thus obtaining a collection of $n = \size{V}$ one-vertex
           components, such that for every $v\in V$ the one-vertex
           component $\Set{v}$ has $\size{\degr{}{v}}$ half edges
           attached to it.
\item[(2)] We splice the two halves of every edge together, not
           of all the edges at once, but in some ordering $\EdgePerm$
           of the edges that minimizes two measures that depend on the
           edge-boundary degrees of assembled components.
\end{itemize}
A component $A$ is a subset of $V$ and its edge-boundary degree is the
number of edges in $G$ with one endpoint in $A$ and one endpoint in
$V-A$ (which is the same as the number of half edges attached to $A$
after all edges with both endpoints in $A$ have been spliced together).
The \textbf{maximum} edge-boundary degree encountered during the
reassembling process is what we call the
$\bm{\alpha}$\textbf{-measure} of the reassembling, and
the \textbf{sum} of all edge-boundary degrees is its
$\bm{\beta}$\textbf{-measure}. The $\alpha$-optimization
(resp. $\beta$-optimization) of the reassembling of $G$ is to determine an
order $\EdgePerm$ for splicing the edges that minimizes its
$\alpha$-measure (resp. $\beta$-measure).

\medskip 
\noindent
There are different forms of reassembling, depending on restrictions
and variations on the ordering $\EdgePerm$ 
of the edges.  We consider only cases satisfying the condition that if the an
edge between disjoint components $A$ and $B$ is spliced, then all the
edges between $A$ and $B$ are spliced at the same time.  In this
report, we examine the particular case of \emph{linear reassembling},
which requires that the next edge to be spliced must be adjacent to an
already spliced edge. We delay other forms of reassembling to
follow-up reports.

\medskip 
\noindent
We prove that $\alpha$-optimization of linear reassembling and
\emph{minimum-cutwidth linear arrangment} ($\mincut$) 
are polynomially reducible to each other, and 
that $\beta$-optimization of linear reassembling and
\emph{minimum-cost linear arrangement} ($\ola$) 
are polynomially reducible to each other. The known NP-hardness of
$\mincut$ and $\ola$ imply the NP-hardness of $\alpha$-optimization
and $\beta$-optimization. 


  \end{abstract}

\ifTR
    \newpage
    \tableofcontents
    \newpage
    \pagenumbering{arabic}
\else
    \vspace{-.2in}
\fi

\section{Introduction}
\label{sect:intro}

We start with a gentle presentation of our graph problem
and then explain the background that motivates our examination.
\paragraph{Problem Statement.}
Let $G = (V,E)$ be a simple (no self-loops and no multi-edges),
connected, undirected graph, with $\size{V} = n\geqslant 1$ vertices
and $\size{E} = m$ edges.  One version of the \emph{reassembling} of
$G$ is edge-directed and can be defined by a \emph{total order}
$\EdgePerm$ of the $m$ edges of $G$.  Informally and very simply, a
total order $\EdgePerm$ of the edges gives rise to a reassembling of
$G$ as follows:
\begin{itemize}[itemsep=0pt,parsep=2pt,topsep=5pt,partopsep=0pt] 
\item[(1)]  We cut every edge into two halves, 
            thus obtaining a collection of $n$ disconnected one-vertex
            components, such that for every $v\in V$ the one-vertex
            component $\Set{v}$ has $\size{\degr{}{v}}$ half edges
            attached to it.
\item[(2)]  We reconnect the two halves of every edge in the order 
            specified by $\EdgePerm$, obtaining larger and larger
            components, until the original $G$ is fully reassembled.
\end{itemize}
To distinguish this reassembling of $G$ according to an order
$\EdgePerm$ from a later reassembling of $G$ more suitable for
parallel computation, we call the former
\emph{sequential reassembling} and the latter  
\emph{binary reassembling}. 

A \emph{bridge} is a yet-to-be-reconnected edge between two components, say,
$A$ and $B$, with disjoint sets of vertices; we call such components
\emph{clusters}.%
   \footnote{These terms (\emph{bridge}, \emph{cluster}, and others, 
   later in this report) are overloaded in graph-theoretical problems. 
   We make our own use of these terms, and state it explicitly when
   our meaning is somewhat 
   at variance with that elsewhere in the literature.} 
The set of bridges
between $A$ and $B$ is denoted $\bridges{}{A,B}$. For technical
reasons, when we reconnect one of the bridges in $\bridges{}{A,B}$, we also
reconnect all the other bridges in $\bridges{}{A,B}$ and cross them out
from further consideration in $\EdgePerm$. Thus, in the
reassembling of $G$ according to $\EdgePerm$, there are at most $m$
steps, rather than exactly $m$ steps.

In the case when $B = V-A$, the set $\bridges{}{A,B}$ is the same as  
the cut-set of edges determined by the cut $(A,V-A)$.
Instead of $\bridges{}{A,V-A}$, we write $\bridges{}{A}$.
The \emph{edge-boundary degree} of a cluster $A$ is the number of
bridges with only one endpoint in $A$, \ie, $\size{\bridges{}{A}}$.

Several natural optimization problems can be associated with graph
reassembling. Two such optimizations are the following, which we
identify by the letters $\alpha$ and $\beta$ throughout:
\begin{itemize}[itemsep=0pt,parsep=2pt,topsep=5pt,partopsep=0pt] 
\item[($\alpha$)] 
    Minimize the \textbf{maximum} edge-boundary degree encountered during
    reassembling.
\item[($\beta$)]
    Minimize the \textbf{sum} of all edge-boundary degrees encountered during
    reassembling.
\end{itemize}
Initially, before we start reassembling, we always set the $\alpha$-measure
$M_{\alpha}$ to the \textbf{maximum} of all the vertex degrees, \ie, 
$\max\,\Set{\degr{}{v}\,|\,v\in V}$,
and we set the $\beta$-measure $M_{\beta}$ 
to the \textbf{sum} of the vertex degrees, \ie,
$\sum\,\Set{\degr{}{v}\,|\,v\in V}$, regardless of which strategy,
\ie, the order $\EdgePerm$ of edges, is selected for the reassembling.
During reassembling, after we merge disjoint nonempty clusters $A$
and $B$, we update the $\alpha$-measure $M_{\alpha}$ to:
$\max\,\SET{M_{\alpha},\size{\bridges{}{A\cup B}}}$, and the
$\beta$-measure $M_{\beta}$ to: $\bigl(M_{\beta}
+ \size{\bridges{}{A\cup B}}\bigr)$. The reassembling process
terminates when only one cluster remains, which is also the set $V$ of
all the vertices.

In what we call the \emph{binary reassembling} of $G$, we reconnect
bridges in several non-overlapping pairs of clusters
simultaneously. That is, at every step -- which we may call
a \emph{parallel step} for emphasis -- we choose $k\geqslant 1$ and
choose $k$ cluster pairs $(A_1,B_1), \ldots, (A_k,B_k)$, where
$A_1,B_1, \ldots, A_k,B_k$ are $2k$ pairwise disjoint clusters (\ie,
with pairwise disjoint subsets of vertices), and simultaneously
reconnect all the bridges in $\sum_{1\leqslant i\leqslant
k}\bridges{}{A_i,B_i}$.  The subsets $A_1,B_1, \ldots, A_k,B_k$ may or
may not include all of the vertices, \ie, in general 
$\sum_{1\leqslant i\leqslant k} (A_i\uplus B_i) \subseteq V$ rather than $= V$.  
(We write ``$\uplus$'' to denote \emph{disjoint union}.)

A binary reassembling is naturally viewed as vertex-directed and
described by a binary tree $\B$ -- root at the top, leaves at the
bottom -- with $n$ leaf nodes, one for each of
the initial one-vertex clusters.  Each non-leaf node in $\B$ is a
cluster $(A\uplus B)$ obtained by reconnecting all the bridges in
$\bridges{}{A,B}$ between the sibling clusters $A$ and $B$.  The first
parallel step, $i=0$, starts at the bottom in the reassembling
process, by considering the $n$ leaf nodes of $\B$ and calculating the
max (for $\alpha$ optimization) or the sum (for $\beta$ optimization)
of all vertex degrees. If $h$ is the height of $\B$, the last parallel
step is $i=h$, which corresponds to the root node of $\B$ (the entire
set $V$ of vertices) and produces the final $\alpha$-measure and
$\beta$-measure.  Clearly, 
$\lceil\log n\rceil \leqslant h\leqslant n-1$.
 
Every sequential reassembling of $G$ can be viewed as a binary
reassembling of $G$ where, at every step, only one cluster pair
$(A,B)$ is selected and one nonempty set of bridges $\bridges{}{A,B}$
is reconnected. Conversely, by serializing (or sequencializing) parallel
steps, every binary reassembling which we call \emph{strict} can be
re-defined as a sequential reassembling. Details of the correspondence
between sequential and binary reassemblings are in 
Appendix~\ref{sect:sequential}.

A binary reassembling is \emph{strict} if the merging of a cluster
pair $(A,B)$ is restricted to the case 
$\bridges{}{A,B} \neq \varnothing$. If an $\alpha$-optimal
(resp. $\beta$-optimal) binary reassembling is strict, then its
serialization is an $\alpha$-optimal (resp. $\beta$-optimal)
sequential reassembling.

\paragraph{The Linear Case.}

A possible and natural variation (or restriction) of graph
reassembling is one which we call \emph{linear}. If, at every step of
the reassembling process, we require that the cluster pair $(A,B)$ to
be merged is such that one of the two clusters, $A$ or $B$ (or both at
the first step), is a singleton set, then the resulting reassembling
is \emph{linear}.  The binary tree $\B$ describing a linear
reassembling of a graph $G$ with $n$ vertices is therefore a
degenerate tree of height $h = n-1$.

Clearly, there can be no non-trivial parallel step which merges two
(or more) cluster pairs in linear reassembling, \ie, no step which
simultaneously merges disjoint cluster pairs $(A_1,B_1)$ and
$(A_2,B_2)$ to form the non-singleton clusters $(A_1\uplus B_1)$ and
$(A_2\uplus B_2)$, before they are merged to form the cluster
$\bigl((A_1\uplus B_1) \uplus (A_2\uplus B_2)\bigr)$.  

Another useful way of understanding a linear reassembling of graph $G$
is in its sequential formulation (when the reassembling is strict):
The order $\EdgePerm$ for reconnecting the edges is such that the next
edge to be reconnected is always adjacent to an already re-reconnected
edge, which enforces the requirement that the next cluster pair
$(A,B)$ to be merged is always such that $A$ or $B$ is a singleton
set.

There are other natural variations of graph reassembling, such as
\emph{balanced reassembling}, whose binary tree $\B$
maximizes the merging of cluster pairs at every parallel step and
whose height $h$ is therefore $\lceil\log n\rceil$. We study these
in follow-up reports.

\paragraph{Background and Motivation.}
Besides questions of optimization and the variations which it
naturally suggests, graph \emph{reassembling} (and the related operation of
graph \emph{assembling}, not considered in this paper) is part of the
execution by programs in a domain-specific language (DSL) for the
design of flow networks
\cite{BestKfoury:dsl11,Kfoury:sblp11,Kfoury:SCP2014,%
SouleBestKfouryLapets:eoolt11}. 
In network \emph{reassembling}, the network is taken apart 
and reassembled in an order determined by the designer;
in network \emph{assembling}, the order in which components are put together
is pre-determined, which is the order in which components become
available to the designer. 

A flow network is a directed graph where vertices and edges are assigned 
various attributes that regulate flow through the network.%
   \footnote{Such networks are typically more complex 
             than the capacited directed graphs that algorithms
             for max-flow (and other related quantities) and its
             generalizations (\eg, multicommodity max-flow) operate on.}
Programs for flow-network design are meant to connect network
components in such a way that \emph{typings at their interfaces},
\ie, formally specified properties at their common boundaries, 
are satisfied. Network typings guarantee there are no conflicting
data types when different components are connected, and insure that
desirable properties of safe and secure operation are not violated by
these connections, \ie, they are \emph{invariant properties} of the
whole network construction.

A typing $\tau$ for a \emph{network component} $A$ (or \emph{cluster}
$A$ in this report's terminology) formally expresses a constraining
relationship between the variables denoting the outer ports of $A$ (or
the edge-boundary $\bridges{}{A}$ in this report). The smaller the set
of outer ports of $A$ is, the easier it is to formulate the typing
$\tau$ and to test whether it is compatible with the typing $\tau'$ of
another network component $A'$. Although every outer port of $A$ is
directed, as input port or output port, the complexity of the
formulation of $\tau$ depends only on the number of outer ports (or
$\size{\bridges{}{A}}$ in this report), not on their directions. 

If $k$ is a uniform upper bound on the number of outer ports of all
network components, the time complexity of reassembling the network
without violating any component typing $\tau$ can be made linear in
the size $n$ of the completed network and exponential in the bound $k$
-- not counting the preprocessing time $f(n)$ to determine an
appropriate reassembling order. Hence, the smaller are $k$ and $f(n)$,
the more efficient is the construction of the entire network. From
this follows the importance of minimizing the preprocessing time
$f(n)$ for finding a reassembling strategy that also minimizes the
bound $k$.

\paragraph{Main Results.}
In this report, we restrict attention to the linear case of
graph reassembling, which is interesting and natural in its own right. 
We first prove that $\alpha$-optimization and $\beta$-optimization
of linear reassembling are both NP-hard problems. We obtain these results
by showing that:
\begin{itemize}[itemsep=0pt,parsep=2pt,topsep=5pt,partopsep=0pt] 
\item
   $\alpha$-optimization of linear reassembling and 
   \emph{minimum-cutwidth linear arrangment} ($\mincut$) \\
   are polynomially-reducible to each other,
\item
   $\beta$-optimization of linear reassembling and 
   \emph{minimum-cost linear arrangement} ($\ola$) \\
   are polynomially-reducible to each other.
\end{itemize}
Both $\mincut$ and $\ola$ have been extensively studied: They
are both NP-hard in general, but, in a few non-trivial special cases,
amenable to low-degree polynomial-time solutions~\cite{petit:2011}. 
By our polynomial-reducibility results, these polynomial-time solutions
are transferable to the $\alpha$-optimization and $\beta$-optimization
of linear reassembling. This leaves open the problem of identifying 
classes of graphs, whether of practical or theoretical significance,
for which there are low-degree polynomial-time solutions for
our two optimization problems.

\paragraph{Organization of the Report.}

In Section~\ref{sect:problem} we give precise formal definitions of
several notions underlying our entire examination. The formulation of
some of these (\eg, our definition of binary trees) is not standard,
but which we purposely choose in order to facilitate the subsequent
analysis.

In Section~\ref{sect:examples} we give several examples to 
illustrate notions discussed in this Introduction, in 
Section~\ref{sect:problem}, as well as in later sections.

In Sections~\ref{sect:alpha-optimization} and~\ref{sect:beta-optimization}
we prove our NP-hardness results about $\alpha$-optimization and
$\beta$-optimization. Some of the long technical proofs for these
sections are delayed to Appendix~\ref{sect:remaining-proofs-for-linear}. 

In the final short Section~\ref{sect:future}, we point to open 
problems and to further current research on 
$\alpha$-optimization and
$\beta$-optimization of graph reassembling.



\section{Formal Definitions and Preliminary Lemmas}
\label{sect:problem}

Let $G = (V,E)$ be a simple (no self-loops and no multi-edges)
undirected graph, with $\size{V} = n\geqslant 1$ and
$\size{E} = m \geqslant 1$.  Throughout, there is no loss of
generality in assuming that $G$ is connected.

As pointed out in Section~\ref{sect:intro},
there are two distinct definitions of the reassembling of
$G$. The first is easier to state informally: This is the 
definition of \emph{sequential reassembling}.
The second definition, \emph{binary reassembling},
is more convenient for the optimization problems we want to study.
The analysis in this report and follow-up reports is based
on the formal definition of binary reassembling and its variations;
we delay a formal definition of sequential reassembling
to Appendix~\ref{sect:sequential}, where
we also sketch the proof of the equivalence of the two definitions
when reassembling is \emph{strict} 
(see Definition~\ref{rem:strict-binary-reassembling}).

\subsection{Binary Graph-Reassembling}
\label{sect:binary-reassembling}

Our notion of a \emph{binary reassembling} of graph $G$ presupposes
the notion of a \emph{binary tree} over a finite set
$V=\Set{v_1,\ldots,v_n}$.  Our definition of \emph{binary trees} is
not standard, but is more convenient for our purposes.%
   \footnote{A standard definition of a binary tree $T$ makes
   $T$ a subset of the set of finite binary strings $\Set{0,1}^*$ such that:
\begin{itemize}[itemsep=0pt,parsep=0pt,topsep=2pt,partopsep=0pt] 
\item $T$ is prefix-closed, \ie, if $t\in T$ and $u$ is a prefix of $t$,
      then $u\in T$.
\item Every node has two children, \ie, $t\,0\in T$ iff $t\,1\in T$
      for every $t\in \Set{0,1}^*$.
\end{itemize}
The \emph{root node} of $T$ is the empty string $\varepsilon$,
and a \emph{leaf node} of $T$ is a string $t\in T$ without children, \ie,
both $t\,0\not\in T$ and $t\,1\not\in T$. 
}

\begin{definition}{Binary trees}
\label{defn:binaryTrees}
An \emph{(unordered) binary tree} $\B$ over $V=\Set{v_1,\ldots,v_n}$
is a collection of non-empty subsets of $V$ satisfying three
conditions:
\begin{enumerate}[itemsep=0pt,parsep=2pt,topsep=5pt,partopsep=0pt] 
\item For every $v\in V$, the singleton set $\Set{v}$ is in $\B$.
\item The full set $V$ is in $\B$.
\item For every $X\in\B$, there is a unique $Y\in\B$ such that:
      $X\cap Y=\varnothing$ and $(X\cup Y)\in\B$.
\end{enumerate}
The \emph{leaf nodes} of $\B$ are the singleton sets $\Set{v}$, and
the \emph{root node} of $\B$ is the full set $V$. Depending on the
context, we may refer to the members of $\B$ as its \emph{nodes} or as
its \emph{clusters}. Several expected properties of $\B$, reproducing
familiar ones of a standard definition, are stated in the next two
propositions.
\end{definition}

\begin{proposition}[Properties of binary trees]
\label{prop:propertiesOne}
Let $\B$ be a binary tree as in Definition~\ref{defn:binaryTrees}, let $v\in V$,
and let:
\begin{equation*}
\label{eq:path} 
\tag{$\dag$}
    \Set{v} = X_0\ \subsetneq\ X_1 \ \subsetneq X_2 
                  \ \subsetneq\ \cdots\ \subsetneq\ X_p = V
\end{equation*}
be a maximal sequence of nested clusters from $\B$. We then have:
\begin{enumerate}[itemsep=0pt,parsep=2pt,topsep=5pt,partopsep=0pt] 
\item The sequence in~\eqref{eq:path} is uniquely defined, \ie,
      every maximal nested sequence starting from $\Set{v}$ is the same.
\item For every cluster $Y\in\B$, if $\Set{v}\subseteq Y$, then
      $Y\in\Set{X_0,\ldots,X_p}$.
\item There are pairwise disjoint clusters 
      $\Set{Y_0,\ldots,Y_{p-1}} \subseteq \B$ such that, for every 
      $0 \leqslant i < p$:
\[    X_i \cap Y_{i} = \varnothing \quad\text{and}\quad
      X_i \cup Y_{i} = X_{i+1}.
\]
\end{enumerate}
\end{proposition}

\noindent
Based on this proposition, we use the following terminology:
\begin{itemize}[itemsep=0pt,parsep=2pt,topsep=5pt,partopsep=0pt] 
\item We call the sequence in~\eqref{eq:path}, which is unique by part 1,
      the \emph{path} from the leaf node $\Set{v}$ to the root node $V$. 
\item Every cluster containing $v$ is one of the nodes along this unique 
      path, according to part 2.
\item In part 3, we call $Y_i$ the unique \emph{sibling} of 
      $X_i$, whose unique common \emph{parent} is $X_{i+1}$, for every 
      $0 \leqslant i < p$.
\end{itemize}

\begin{proof}
Part 1 is a consequence of the third condition in 
Definition~\ref{defn:binaryTrees}, which prevents any cluster/node
in $\B$ from having two distinct parents. 

For part 2, consider the nested sequence $\Set{v}\subseteq Y\subseteq
V$, and extend it to a maximal nested sequence, which is uniquely defined
by part 1. This implies $Y$ is one of the nodes along the path
from $\Set{v}$ to the root $V$.

Part 3 is another consequence of the third condition in 
Definition~\ref{defn:binaryTrees}: For every $0 \leqslant i < p$,
$Y_i$ is the unique cluster in $\B$ such that $X_i\cap Y_i =\varnothing$
and $X_i\cup Y_i\in\B$. Remaining details omitted.
\end{proof}

\begin{proposition}[Properties of binary trees]
\label{prop:propertiesTwo}
Let $\B$ be a binary tree as in Definition~\ref{defn:binaryTrees}.
We then have:
\begin{enumerate}[itemsep=0pt,parsep=2pt,topsep=5pt,partopsep=0pt] 
\item For all clusters $X,Y\in\B$,
      if $X\cap Y\neq\varnothing$ then $X\subseteq Y$
      or $Y\subseteq X$.
\item For every cluster $X\in\B$, the sub-collection
      of clusters ${\B}_{X} := \Set{\,Y\in\B\;|\;Y\subseteq X\,}$ is a
      binary tree over $X$, with root node $X$.
\item $\B$ is a collection of $(2n-1)$ clusters.
\end{enumerate}
\end{proposition}

\begin{proof}
For part 1, let $Z = X\cap Y \neq \varnothing$ and consider
the maximal sequence of nested clusters which 
extends $Z\subseteq X\subseteq V$ or $Z\subseteq Y\subseteq V$.
By part 1 of Proposition~\ref{prop:propertiesOne}, both $X$ and
$Y$ must occur along this nested sequence.

For part 2, we directly check that $\B_X$ satisfies the three defining
properties of a binary tree. Straightforward details omitted.

Part 3 is proved by induction on $n\geqslant 1$. For the induction
hypothesis, we assume the statement is true for every binary tree 
with less than $n$ leaf nodes, and we then prove that the statement is true
for an arbitrary binary tree over $V$ with $\size{V} = n$. Consider
a largest cluster $X\in\B$ such that $X\neq V$. There is a unique
$Y\in\B$ such that $X\cap Y=\varnothing$ and $X\cup Y\in\B$.
Because $X$ is largest in size but smaller than $V$, it must be that
$X\cup Y = V$, which implies that $\Set{X,Y}$ is a two-block partition
of $V$. Consider now the subtrees ${\B}_X$ and ${\B}_Y$, apply the
induction hypothesis to each, and draw th desired conclusion. 
\end{proof}

A common measure for a standard definition of binary trees is $\heightSym$,
which becomes for our notion of binary trees in 
Definition~\ref{defn:binaryTrees}:
\begin{alignat*}{5}
    &\height{}{\B}\ &&:=
    \ && \max\,\bigl\{\,p\;\bigl|\;&& \text{there is $v\in V$ such that }
\\
    & && && && 
      \Set{v} = X_0 \subsetneq X_1 \subsetneq\cdots\subsetneq X_p = V
      \text{ is a maximal sequence of nested clusters}\,\bigr\} .
\end{alignat*}
For a particular node/cluster $X \in\B$, the subtree of $\B$ rooted at
$X$ is ${\B}_X$, by part 2 in Proposition~\ref{prop:propertiesTwo}.
The height of $X$ in $\B$ is therefore
$\height{\B}{X} := \height{}{{\B}_X}$.

If $X,Y\subseteq V$ are disjoint sets of vertices, $\bridges{G}{X,Y}$ is
the subset of edges of $G$ with one endpoint in each of $X$ and $Y$.
If $Y = V - X$, we write $\bridges{G}{X}$ instead of $\bridges{G}{X,V-X}$.

\begin{definition}{Binary reassembling}
\label{defn:binary-reassembling}
A \emph{binary reassembling of the graph $G = (V,E)$} is simply
defined by a pair $(G,\B)$ where $\B$ is a binary tree over $V$, as in 
Definition~\ref{defn:binaryTrees}.

Given a binary reassembling $(G,\B)$ of $G$, two measures are of particular
interest for our later analysis, namely, for every cluster $X\in\B$, the
\emph{degree} of $X$ and the \emph{height} of $X$:
\[ 
  \degr{G,\B}{X}\ :=\ \size{\bridges{G}{X}}
\quad\text{and}\quad
  \height{G,\B}{X}\ :=\ \height{\B}{X} .
\] 
If the context makes clear the binary reassembling $(G,\B)$ -- respectively,
the binary tree $\B$ -- relative
to which these measures are defined, we write $\degr{}{X}$ and
$\height{}{X}$ -- respectively,  $\degr{G}{X}$ and $\height{G}{X}$ -- 
instead of $\degr{G,\B}{X}$ and $\height{G,\B}{X}$.%
     \footnote{ Our binary reassembling of $G$ can be viewed
        as the ``hierarchical clustering'' of $G$, similar to a method of
        analysis in data mining, though used for a different purpose. Our
        binary reassembling mimicks what is called ``agglomerative, or
        bottom-up, hierarchical clustering'' in data mining.}
\end{definition}

\begin{definition}{Strict binary reassembling}
\label{defn:strict-binary-reassembling}
\label{rem:strict-binary-reassembling}
We say the binary reassembling $(G,\B)$ in
Definition~\ref{defn:binary-reassembling} is \emph{strict} if 
it satisfies the following condition:
For all clusters $X,Y\in\B$, if $X$ and $Y$ are sibling nodes,
then $\bridges{G}{X,Y} \neq \varnothing$.
(In Definition~\ref{defn:binary-reassembling}, when we merge 
sibling clusters $X,Y\in\B$, we do not require that
$\bridges{G}{X,Y} \neq \varnothing$.)
\end{definition}

\subsection{Optimization Problems}
\label{sect:optimization-problems}

The following definition repeats a definition 
in Section~\ref{sect:intro} more formally.

\begin{definition}{Measures on the reassembling of a graph}
\label{def:different-measures}
Let $(G,\B)$ be a binary reassembling of $G$.
We define the measures $\alpha$ and $\beta$
on $(G,\B)$ as follows:
\Hide{
We define the measures $\alpha_i$ and $\beta_i$
on $(G,\B)$ according to height $i$, 
for every $0\leqslant i\leqslant k = \height{}{\B}$, as follows:
\begin{alignat*}{10}
    & \alpha_i(G,\B) &&:= 
     \ &&\max\, &&\SET{\,\degr{G,\B}{X}\; \bigl| \; X\in\B \text{ and }
     \height{\B}{X} = i\,},
\\[1.2ex]
    & \beta_i(G,\B) &&:= 
     \ &&\sum\, &&\SET{\,\degr{G,\B}{X}\; \bigl| \; X\in\B \text{ and }
     \height{\B}{X} = i\,}.
\end{alignat*}
Based on the preceding, we define the following measures:
\begin{alignat*}{10}
   & \alpha(G,\B)\ &&:=\ &&\max\ &&\SET{\alpha_0(G,\B),\ldots,\alpha_{k}(G,\B)},
\\[1.2ex]
   & \beta(G,\B)\ &&:=\ &&\sum\ &&\SET{\beta_0(G,\B),\ldots,\beta_{k}(G,\B)}.
\end{alignat*}
}
\begin{alignat*}{10}
    & \alpha(G,\B) &&:= 
     \ &&\max\, &&\SET{\,\degr{G,\B}{X}\; \bigl| \; X\in\B \,},
\\[1.2ex]
    & \beta(G,\B) &&:= 
     \ &&\sum\, &&\SET{\,\degr{G,\B}{X}\; \bigl| \; X\in\B \,}.
\end{alignat*}
An optimization problem arises with the minimization of each of
these measures. For example, the optimal $\alpha$-measure of graph $G$ is:
\[
  \alpha (G)\:=\ \min\,
  \SET{\,\alpha (G,\B)\;\bigl|\;(G,\B)\text{ is a binary reassembling of $G$}\,} .
\]
We say the binary reassembling $(G,\B)$ is \emph{$\alpha$-optimal}
iff $\alpha(G,\B) = \alpha(G)$. We leave to the reader
the obvious similar definition of what it means
for the binary reassembling $(G,\B)$ to be \emph{$\beta$-optimal}.
\end{definition}

\subsection{Linear Graph-Reassembling}
\label{sect:linear-reassembling}

Let $G = (V,E)$ be a simple undirected graph with $\size{V} = n$.  We
say the binary reassembling $(G,\B)$ is \emph{linear} if $\B$ is
a \emph{linear binary tree} over $V$, \ie, all the clusters of size 
$\geqslant 2$ forms a single nested chain of length $(n-1)$:
\begin{equation*}
\label{eq:linear1} 
\tag{A}
    X_1 \ \subsetneq X_2 
                  \ \subsetneq\ \cdots\ \subsetneq\ X_{n-1} = V
\end{equation*} 
This implies the height of $\B$ is
$(n-1)$. By part 3 in Proposition~\ref{prop:propertiesOne}, there are
$n$ leaf nodes/singleton clusters $\Set{Y_0,\ldots,Y_{n-1}} \subseteq \B$ 
such that $X_1 = Y_0\cup Y_1$ and for every $1 \leqslant i \leqslant n-2$:
\begin{equation*}
\label{eq:linear2} 
\tag{B}
      X_i \cap Y_{i+1} = \varnothing \quad\text{and}\quad
      X_i \cup Y_{i+1} = X_{i+1}.
\end{equation*} 
We use the letter $\LL$ 
to denote a linear binary tree, and write $(G,\LL)$ to denote a linear
reassembling of $G$.

In Definition~\ref{def:linear-arrangement}, 
we mostly use the notation and conventions of~\cite{petit:2011} and
the references therein. We write $\set{v\,w}$ to denote the edge
connecting vertex $v$ and vertex $w$.

\begin{definition}{Linear arrangements and cutwidths}
\label{def:linear-arrangement} 
A \emph{linear arrangement} $\varphi$ of the graph 
$G = (V,E)$, where $\size{V} = n$, is a bijection 
$\varphi: V\to\Set{1,\ldots,n}$.  We refer to this linear
arrangement by writing $(G,\varphi)$. 

Following~\cite{petit:2011},
given linear arrangement $(G,\varphi)$ and $i\in\Set{1,\ldots,n}$, we
define a two-block partition of the vertices, $V =
L(G,\varphi,i)\uplus R(G,\varphi,i)$, where:
\[
    L(G,\varphi,i) := \Set{\,v\in V\;|\;\varphi(v)\leqslant i\,}
    \quad\text{and}\quad
    R(G,\varphi,i) := \Set{\,w\in V\;|\;\varphi(w) > i\,} .
\]
The \emph{edge cut at position $i$}, denoted $\zeta(G,\varphi,i)$, is
the number of edges connecting $L(G,\varphi,i)$ and $R(G,\varphi,i)$:
\[
    \zeta(G,\varphi,i) := \ssize{\Set{\,\set{v\,w}\in E\;|
          \; v\in L(G,\varphi,i) \text{ and }
             w\in R(G,\varphi,i) \,}} .
\]
In our notation in Definition~\ref{defn:binary-reassembling}, we have:
\[
    \zeta(G,\varphi,i)\ =\ \ssize{\Bridges{}{L(G,\varphi,i),R(G,\varphi,i)}}
    \ =\ \Degr{}{L(G,\varphi,i)} .
\]
The \emph{length of\ $\set{v\,w}$\ in the linear arrangement $(G,\varphi)$}, 
denoted $\xi(G,\varphi,\set{v\,w})$, is 
``$1$ + the number of vertices between $v$ and $w$'':
\[
    \xi(G,\varphi,\set{v\,w}) := \ssize{\varphi(v)-\varphi(w)} .
\]
The $\alpha$-measure and $\beta$-measure of the linear arrangement
$(G,\varphi)$ are defined by:
\begin{alignat*}{9}
  & \alpha(G,\varphi)\ &&:=
  \ && \max\, && \Set{\,\zeta(G,\varphi,i)\;|\;1\leqslant i\leqslant n\;}
  \ &&= \ && 
  \max\, && \SET{\,\Degr{}{L(G,\varphi,i)}\;\bigl|\;1\leqslant i\leqslant n\;},
\\
  & \beta(G,\varphi)\ &&:=
  \ && \sum\, && \Set{\,\zeta(G,\varphi,i)\;|\;1\leqslant i\leqslant n\;}
  \ &&= \ && 
  \sum\, && \SET{\,\Degr{}{L(G,\varphi,i)}\;\bigl|\;1\leqslant i\leqslant n\;}.
\end{alignat*}
In the literature, $\alpha(G,\varphi)$ is called the \emph{cutwidth}
of the linear arrangement $(G,\varphi)$. The \emph{cost} of 
the linear arrangement $(G,\varphi)$ is usually defined as  the
total length of all the edges relative to 
$(G,\varphi)$, \ie, the cost is the measure $\gamma(G,\varphi)$ given by:
\[
   \gamma(G,\varphi)\ := 
   \ \sum\Set{\,\xi(G,\varphi,\set{v\,w})\,|\,\set{v\,w}\in E\,}.
\]
However, by Lemma~\ref{lem:equality-of-measures} below, 
$\beta(G,\varphi)$ is equal to $\gamma(G,\varphi)$.
\end{definition}

\begin{lemma}
\label{lem:equality-of-measures}
For every linear arrangement $(G,\varphi)$, we have
$\beta(G,\varphi) = \gamma(G,\varphi)$.
\end{lemma}

\begin{proof}
We have to prove that
\[
    \sum\, \Set{\,\zeta(G,\varphi,i)\;|\;1\leqslant i\leqslant n\;}
    \ =
    \ \sum\, \Set{\,\xi(G,\varphi,\set{v\,w})\;|\;\set{v\,w}\in E\;} .
\]
This equality holds whether $G$ is connected or not. So, a
formal proof (omitted) can be written by induction on the number
$m\geqslant 0$ of edges in $G$.  But informally, for every edge
$\set{v\,w}\in E$, if $\varphi(v) = i$ and $\varphi(w) = j$ with
$i<j$, then its length $\xi(G,\varphi,\set{v\,w}) = j-i$. In this
case, the length of edge $\set{v\,w}$ contributes one unit to each of
$(j-i)$ consecutive edge cuts:
$\zeta(G,\varphi,i),\ldots,\zeta(G,\varphi,j-1)$. Hence, if we delete
edge $\set{v\,w}$ from the graph, we decrease the two quantities:
\[
       \sum\, \Set{\,\zeta(G,\varphi,i)\;|\;1\leqslant i\leqslant n\;}
       \quad\text{and}\quad
       \sum\, \Set{\,\xi(G,\varphi,\set{v\,w})\;|\;\set{v\,w}\in E\;} .
\]
by exactly the same amount $(j-i)$. The desired conclusion follows.
\end{proof}

\begin{definition}{Optimal linear arrangements}
\label{def:optimal-linear-arrangements}
Let $G = (V,E)$ be a simple undirected graph.
We say the linear arrangement $(G,\varphi)$ is \emph{$\alpha$-optimal} if:
\[
   \alpha(G,\varphi)\ =
   \ \min\,\SET{\,\alpha(G,\varphi')\;\bigl|
   \;(G,\varphi')\text{ is a linear arrangement}\,}.
\]
The \emph{$\alpha$-optimal linear arrangement problem} is the problem of 
defining a bijection $\varphi:V\to\Set{1,\ldots,n}$ such that
$(G,\varphi)$ is $\alpha$-optimal. We define similarly the
\emph{$\beta$-optimal linear arrangement problem}.
\end{definition}

\begin{definition}{Linear arrangement induced by linear reassembling}
\label{def:linear-arrangement-induced}
Let $G = (V,E)$ be a simple undirected graph and $(G,\LL)$ be a linear
reassembling of $G$. Using the notation in~\eqref{eq:linear1}
and~\eqref{eq:linear2} in the opening paragraph of
Section~\ref{sect:linear-reassembling}, the $n$ leaf nodes (or singleton
clusters) of $\LL$ are: $Y_0, Y_1, Y_2,\ldots,Y_{n-1}$. Observe that the order of
the singletons $Y_2,\ldots,Y_{n-1}$ and, therefore, the order of the $(n-2)$ 
vertices in $Y_2\cup\cdots\cup Y_{n-1}$ is uniquely determined by the chain 
in~\eqref{eq:linear1}, but this is not the case for the order in which we 
write $Y_0$ and $Y_1$, \ie, the first cluster $X_1$ in~\eqref{eq:linear1}
is equal to both $Y_0\cup Y_1$ and $Y_1\cup Y_0$.

We want to extract a linear arrangement $(G,\varphi)$ from the linear
reassembling $(G,\LL)$. This is achieved by
defining $\varphi:V\to\Set{1,\ldots,n}$ as follows:
\begin{itemize}[itemsep=0pt,parsep=2pt,topsep=5pt,partopsep=0pt] 
\item Let $Y_0 = \Set{v}$ and $Y_1 = \Set{v'}$. 
      \begin{itemize}[itemsep=0pt,parsep=2pt,topsep=2pt,partopsep=0pt] 
      \item If $\degr{}{v}\leqslant\degr{}{v'}$ then set 
            $\varphi(v) := 1$ and $\varphi(v') := 2$.
      \item If $\degr{}{v} \geqslant \degr{}{v'}$ then set 
            $\varphi(v') := 1$ and $\varphi(v) := 2$.
      \end{itemize}
\item For every $2\leqslant i\leqslant n-1$, if $Y_i = \Set{v}$ 
      then set $\varphi(v) :=\ i+1$.
\end{itemize}
It is possible that $\degr{}{v} = \degr{}{v'}$, in which case $\varphi$
may place $v$ first and $v'$ second, or $v'$ first and $v$ second.
This ambiguity is harmless for our analysis, in that 
it does not affect the $\alpha$-measure and the $\beta$-measure of
the linear arrangement $(G,\varphi)$.

Whether  $\varphi$ places $v$ first and $v'$ second, or 
$v'$ first and $v$ second, we call $(G,\varphi)$ a linear arrangement of 
$G$ \emph{induced} by the linear reassembling $(G,\LL)$. 

Note that, for every $1\leqslant i\leqslant n-1$,
we have the equality $X_i = L(G,\varphi,i)$, 
where $L(G,\varphi,i)$ is the set of vertices at position $i$ and
to the left of it, as in Definition~\ref{def:linear-arrangement}.
\end{definition}

\begin{definition}{Linear reassembling induced by linear arrangement}
\label{def:linear-reassembling-induced}
Let $G = (V,E)$ be a simple undirected graph and let
$(G,\varphi)$ be a linear arrangement of $G$. We extract a 
linear reassembling $(G,\LL)$ from
$(G,\varphi)$. The $n$ leaf nodes/singleton clusters of $\LL$ are:
\[
    \SET{\,\Set{v}\;|\;v\in V\,} = 
    \SET{\,\Set{{\varphi}^{-1}(i)}\;|\;1\leqslant i\leqslant n\,}.
\]
The $(n-1)$ non-leaf nodes/clusters of $\LL$ are:
\begin{alignat*}{8}
  &  X_1\ &&:=\quad &&\Set{{\varphi}^{-1}(1),\ \ &&{\varphi}^{-1}(2)}, 
     \\[1ex]
  &  X_2\ &&:=\ &&
       \Set{{\varphi}^{-1}(1),\ &&{\varphi}^{-1}(2),\ &&{\varphi}^{-1}(3)}, 
     \\
  &  && &&\quad \ldots \quad 
     \\
  & X_{n-1}\ &&:=\ &&\Set{{\varphi}^{-1}(1),\ &&{\varphi}^{-1}(2),
        \ &&{\varphi}^{-1}(3),\ \ldots\ ,\ {\varphi}^{-1}(n)}\ =\ V.
\end{alignat*}
We call $(G,\LL)$ the linear reassembling
of $G$ \emph{induced} by the linear arrangement $(G,\varphi)$.
\end{definition}

\section{Examples}
\label{sect:examples}

We present several simple examples to illustrate notions mentioned in
Sections~\ref{sect:intro} and~\ref{sect:problem}.  We first introduce
a convenient notation for specifying binary trees over a set $V$ of
vertices. Let $2^V$ denote the power set of $V$. We define a binary
operation $\set{\;\X\;\Y\;}$ for all $\X\subseteq 2^V$ and
$\Y \subseteq 2^V$ by:
\[
   \set{\ \X\ \Y\ }\ :=
   \ \bigl(\X\cup\Y\bigr)\ \cup\ \bigl((\bigcup\X)\cup (\bigcup\Y)\bigr)
\] 
The examples below illustrate
how we use this operation ``$\set{\phantom{S}}$''.

We overload the overline notation
``$\set{\phantom{S}}$'' to denote a nonempty subset of $V$,
\ie, $\set{v_1\,\cdots\,v_k}$ means $\Set{v_1,\ldots,v_k}$.
In particular, the edge connecting $v$ and $w$ is denoted by
the two-vertex set $\set{v\,w} = \Set{v,w}$. 

The context will make clear whether we use ``$\set{\phantom{S}}$'' to
refer to a set of subsets of $V$ (in the case of binary trees) or to
just a subset of $V$.

\begin{example}
\label{ex:hypercube-graph}
The hypercube graph $Q_3$ is shown on the left in
Figure~\ref{fig:cube}, and three of its reassemblings are 
shown on the right in Figure~\ref{fig:cube}. The top reassembling
is neither balanced nor linear; the middle one
is \emph{balanced}; and the bottom one is \emph{linear}.

\begin{figure}[hc!]
  \begin{center}
     \includegraphics[width=0.35\textwidth,height=5cm]
          {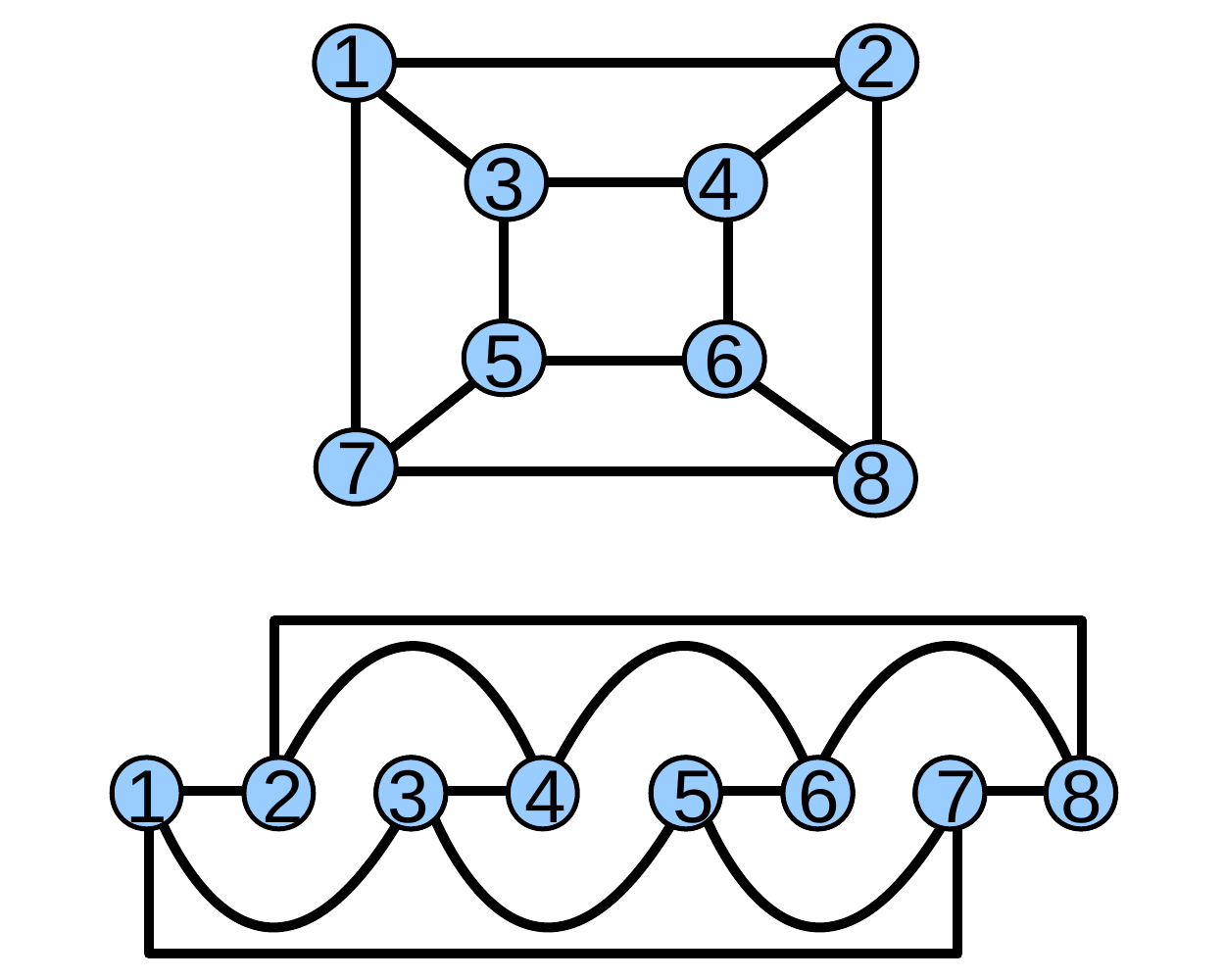}
     \qquad\qquad
     \includegraphics[width=0.35\textwidth,height=5cm]
          {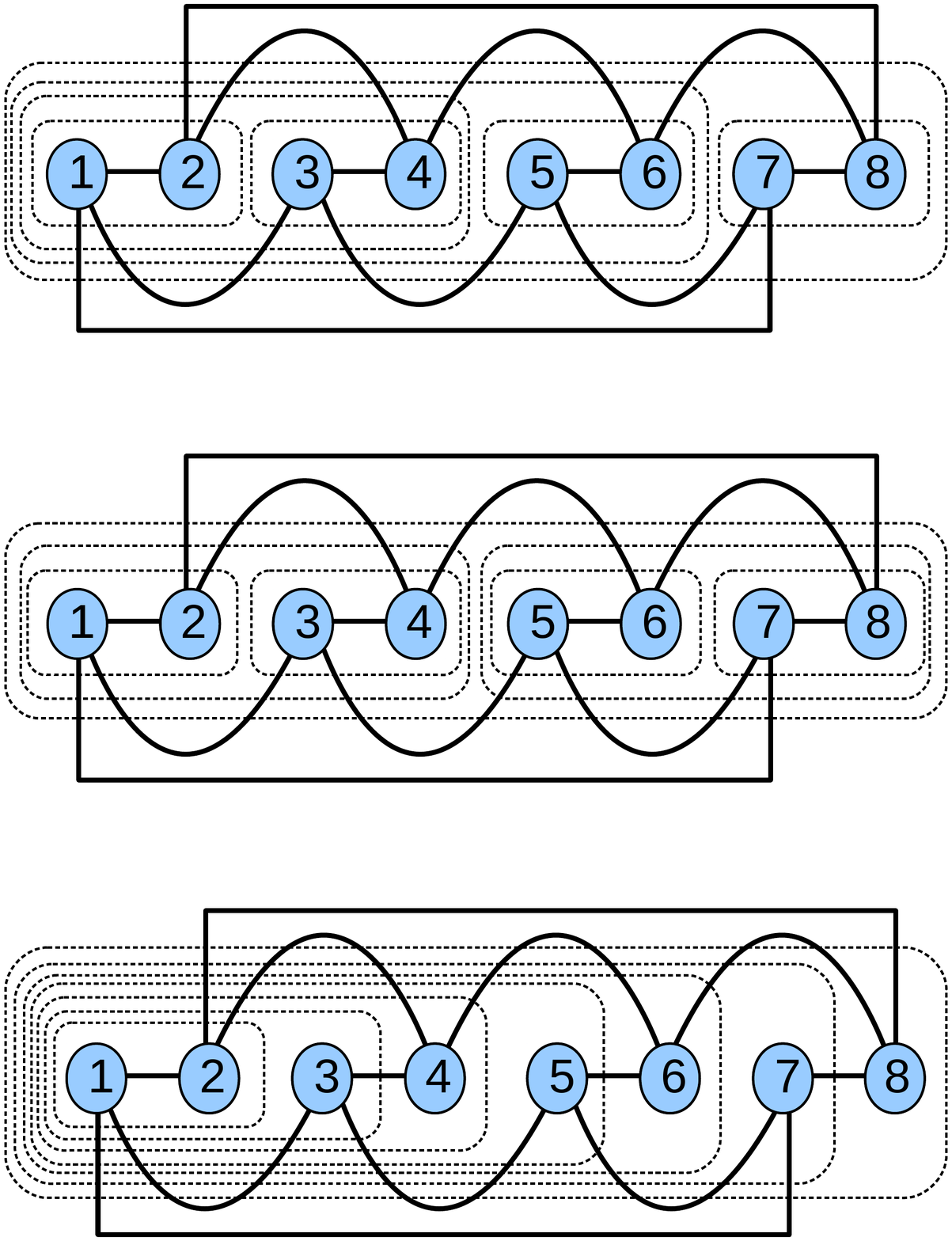}
  \end{center}
     \caption{Two different drawings of the hypercube graph $Q_3$ (on the left)
     and three of its reassemblings (on the right).}
   \label{fig:cube}
\end{figure}

For each of the three reassemblings
on the right in Figure~\ref{fig:cube}, from top to bottom,
we list the unique binary tree $\B$ over the vertices
$\Set{1,\ldots,8}$ underlying it (in its \emph{binary
reassembling} formulation) and one of the orderings ${\EdgePerm}$ 
of the edges $\Set{\set{1\,2},\ldots,\set{7\,8}}$ 
inducing it (in its \emph{sequential reassembling} formulation):

{
\begin{custommargins}{-1.70cm}{0cm} 
\begin{minipage}{0.99\textwidth}
\begin{alignat*}{8}
& {\B}_1\ &&=\quad && 
  \begin{cases}
  \quad \set{\ \set{\ \set{\ \set{\ 1\ \ 2\ }\ \ \set{\ 3\ \  4\ }\ }
   \ \ \ \set{\ 5\ \ 6\ }\ }\ \ \ \set{\ 7\ \ 8\ }\ }
  \end{cases}
     &&  {\EdgePerm}_1^{Q_3} &&=\ && 
     \set{\ 1\ 2\ }\ \ \set{\ 3\ 4\ }\ \ \set{\ 5\ 6\ }\ \ \set{\ 7\ 8\ }
     \ \ \set{\ 1\ 3\ }\ \ \set{\ 3\ 5\ }\ \ \set{\ 5\ 7\ }\ \ \cdots
     \quad && 
\\[3ex] 
& {\B}_2\ &&=\quad && 
     \begin{cases}
     \quad \set{\ \set{\ \set{\ 1\ \ 2\ }\ \ \set{\ 3\ \ 4\ }\ }
     \ \ \ \set{\ \set{\ 5\ \ 6\ }\ \ \set{\ 7\ \ 8\ }\ }\ } 
     \qquad \end{cases}
     &&  {\EdgePerm}_2^{Q_3} &&=\ && 
     \set{\ 1\ 2\ }\ \ \set{\ 3\ 4\ }\ \ \set{\ 5\ 6\ }\ \ \set{\ 7\ 8\ }
     \ \ \set{\ 1\ 3\ }\ \ \set{\ 5\ 7\ }\ \ \set{\ 3\ 5\ }\ \ \cdots
     \quad && \text{(balanced)}
\\[2ex] 
& {\B}_3\ &&=\quad && 
     \begin{cases}
  \quad \set{\ \set{\ \set{\ \set{
  \ \set{\ \set{\ \set{\ 1\ 2\ }\ \ 3\ }\ \ 4\ }\ \ 5\ }\ \ 6\ }\ \ 7\ }\ \ 8\ }
     \qquad \end{cases}
     &&  {\EdgePerm}_3^{Q_3} &&=\ && 
     \set{\ 1\ 2\ }\ \ \set{\ 1\ 3\ }\ \ \set{\ 3\ 4\ }\ \ \set{\ 3\ 5\ }
     \ \ \set{\ 5\ 6\ }\ \ \set{\ 5\ 7\ }\ \ \set{\ 7\ 8\ }\ \ \cdots
     \quad &&  \text{(linear)}
\end{alignat*}
\end{minipage}
\end{custommargins}
}

\medskip\medskip
\noindent
where, for simplicity, we write just ``$v$'' instead of the cumbersome
$\Set{\Set{v}} = \set{\ \set{\ v\ }\ }$.  Thus, for example, two of
the subsets of ${\B}_1$ above appear as ``$\set{\;1\;2\;}$'' and
``$\set{\;\set{\;1\;2\;}\ \set{\;3\;4\;}\;}$'', and if we expand them in full,
we obtain:
\[
   \set{\ 1\ 2\ } = \SET{\Set{1},\Set{2},\Set{1,2}}
   \quad\text{and}\quad
   \set{\ \set{\ 1\ 2\ }\ \ \set{\ 3\ 4\ }\ } = 
   \SET{\Set{1},\Set{2},\Set{3},\Set{4},\Set{1,2},\Set{3,4},\Set{1,2,3,4}}.
\] 
The ellipsis ``$\ldots$'' in the definition
of ${\EdgePerm}_i^{Q_3}$ above are the remaining edges of $Q_3$, which can
be listed in any order without changing the reassembling.%
   \footnote{ We qualify ${\EdgePerm}_i^{Q_3}$ with the superscript
   ``${Q_3}$'' because it depends on the graph ${Q_3}$. The same
   ordering of the edges may not be valid for a sequential
   reassembling of another $8$-vertex graph with a set of edges
   different from that of ${Q_3}$.  This is not the case for the
   binary tree ${\B}$ underlying the binary reassembling $(G,\B)$ of a
   graph $G = (V,E)$; that is, regardless of the placement of edges in
   $G$, the tree ${\B}$ over $V$ is valid for the binary reassembling
   $(G,\B)$ and again for the binary reassembling $(G',\B)$ of every
   graph $G' = (V,E')$ over the same set $V$ of vertices.}
A simple calculation of the $\alpha$-measure and $\beta$-measure of these 
three reassemblings of $Q_3$ produces:
\begin{alignat*}{5}
  & \alpha(Q_3,{\B}_1)\ =\ \alpha(Q_3,{\B}_2)\ =\ 4\quad && \text{and}\quad
    && \alpha(Q_3,{\B}_3)\ =\ 5,
\\
  & \beta(Q_3,{\B}_1)\ =\ \beta(Q_3,{\B}_2)\ =\ 48\quad && \text{and}\quad
    && \beta(Q_3,{\B}_3)\ =\ 49 .
\end{alignat*}

\noindent
By exhaustive inspection (details omitted), $(Q_3,{\B}_1)$ is both
$\alpha$-optimal and $\beta$-optimal for the class of all binary
reassemblings of $Q_3$. Because the $\alpha$-measure and
$\beta$-measure of $(Q_3,{\B}_1)$ and those of $(Q_3,{\B}_2)$ are
equal, $(Q_3,{\B}_2)$ is also both $\alpha$-optimal and $\beta$-optimal
for the class of all binary reassemblings and, therefore, for the 
smaller class of all \emph{balanced} reassemblings of which it is
a member.
By exhaustive inspection again, $(Q_3,{\B}_3)$ is both
$\alpha$-optimal and $\beta$-optimal for the sub-class of
all \emph{linear} reassemblings, but not for the full class
of all binary reassemblings of $Q_3$.
\end{example}

\begin{example}
\label{ex:complete}
The complete graph $K_8$ on $8$ vertices is shown in
Figure~\ref{fig:complete}. We can carry out three reassemblings of
$K_8$ by using the binary trees ${\B}_1$, ${\B}_2$, and ${\B}_3$, from
Example~\ref{ex:hypercube-graph} again.

\medskip
\begin{custommargins}{-1.0cm}{0cm} 
\begin{tabular}{l l}
\begin{minipage}{0.62\textwidth}
   A straightforward
   calculation of the $\alpha$-measure and $\beta$-measure of the
   resulting reassemblings $(K_8,{\B}_1)$, $(K_8,{\B}_2)$, and $(K_8,{\B}_3)$
   produces the following values:
{  
\begin{alignat*}{5}
  & \alpha(K_8,{\B}_1) = \alpha(K_8,{\B}_2) = \alpha(K_8,{\B}_3) =\ 16 ,
\\
  & \beta(K_8,{\B}_1) =\ 132, \quad
    \beta(K_8,{\B}_2) =\ 136, \quad 
    \beta(K_8,{\B}_3)\ =\ 133.
\end{alignat*}
}  
\!\!Because of the symmetries of $K_8$ (``every permutation of the $8$ vertices
produces another graph isomorphic to $K_8$''), all \emph{balanced} reassemblings
are isomorphic and so are all \emph{linear} reassemblings. Hence,
$(K_8,{\B}_2)$ is trivially $\alpha$-optimal and $\beta$-optimal for
the class of all balanced
reassemblings of $K_8$, and $(K_8,{\B}_3)$ is trivially $\alpha$-optimal 
and $\beta$-optimal for the class of all linear reassemblings of $K_8$.
\end{minipage}
&
\begin{minipage}{0.33\textwidth}
   \begin{figure}[H] 
     \begin{center}
       \includegraphics[scale=.35,trim=0cm 2.18cm 0cm 0cm,clip]
       {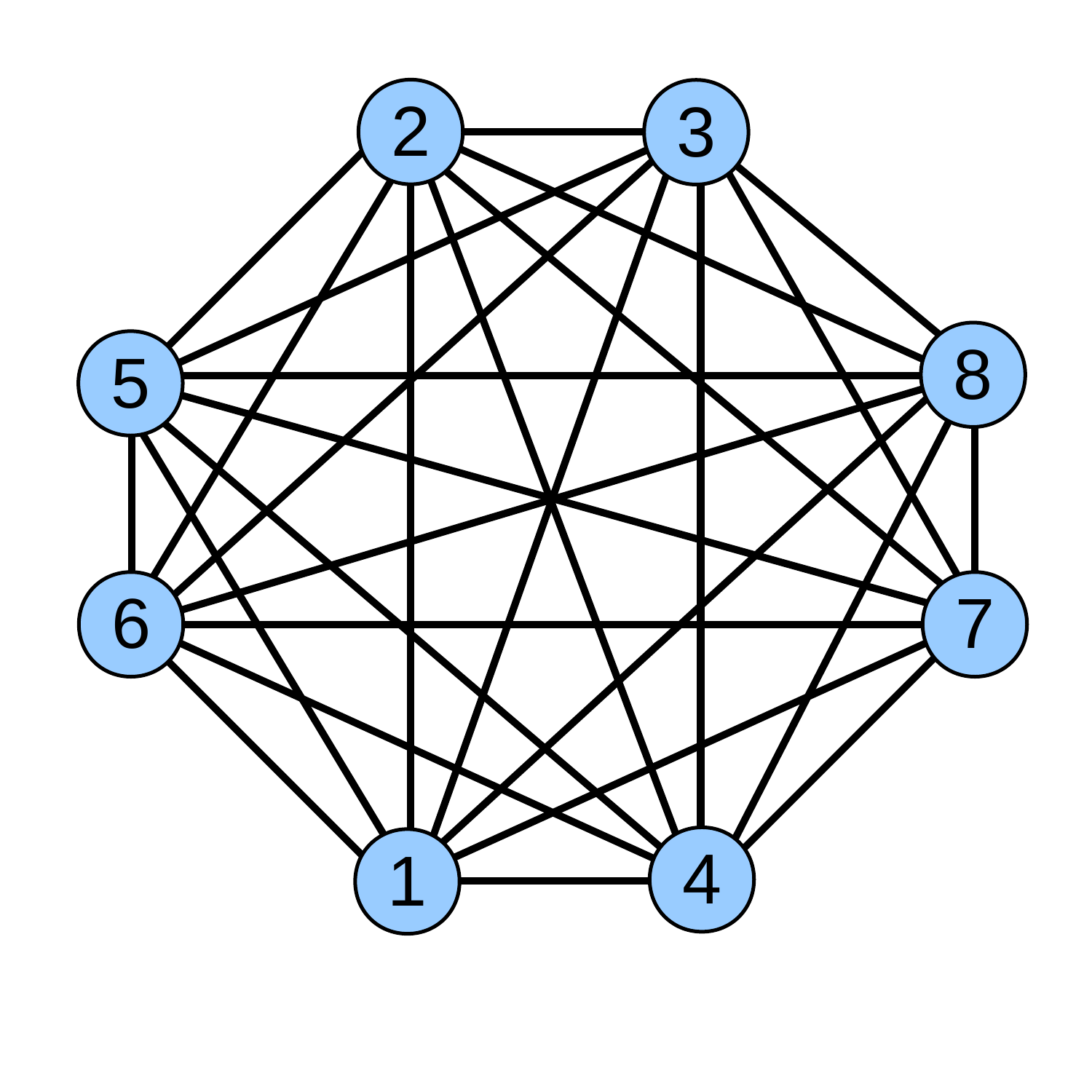}
     \end{center}
      \caption{\footnotesize Complete graph $K_8$.}
      \label{fig:complete}
   \end{figure}
\end{minipage}
\end{tabular}
\end{custommargins}

\smallskip
By exhaustive inspection (details omitted), it turns out that
$(K_8,{\B}_1)$ is $\alpha$-optimal for the class of all binary
reassemblings, but it is not $\beta$-optimal for the same class.  
The underlying tree of a
$\beta$-optimal binary reassembling of $K_8$ turns out to be
the following ${\B}_4$ over the vertices $\Set{1,\ldots,8}$,
shown with an ordering ${\EdgePerm}_4^{K_8}$ 
of the edges which induces a sequential ordering equal to
$(K_8,{\B}_4)$:
\[
  {\B}_4\ =\quad  
  \begin{cases}
  \quad \set{\ \set{\ \set{\ \set{\ \set{\ 1\ \ 2\ }\ \ \set{\ 3\ \  4\ }\ }
   \ \ \ \set{\ 5\ \ 6\ }\ }\ \ \ 7\ }\ \ 8\ }
  \end{cases}
  \qquad
  {\EdgePerm}_4^{K_8} = 
     \set{\ 1\ 2\ }\ \ \set{\ 3\ 4\ }\ \ \set{\ 5\ 6\ }\ \ \set{\ 1\ 3\ }
     \ \ \set{\ 1\ 5\ }\ \ \set{\ 1\ 7\ }\ \ \set{\ 1\ 8\ }\ \ \cdots
\]
where the ellipsis ``$\ldots$'' are the remaining edges in any order.
The resulting $\beta$-measure is $\beta(K_8,{\B}_4) = 127$.
\end{example}

\begin{example}
\label{ex:star}
The star graph $S_7$, with $7$ leaves and one internal vertex,
is shown in Figure~\ref{fig:star}. We can carry out four reassemblings of
$S_7$ by using the binary trees ${\B}_1$, ${\B}_2$, ${\B}_3$, and ${\B}_4$, 
from Examples~\ref{ex:hypercube-graph} and~\ref{ex:complete} again.

\medskip
\begin{custommargins}{-1.0cm}{0cm} 
\begin{tabular}{l l}
\begin{minipage}{0.62\textwidth}
   From a straightforward
   calculation of the $\alpha$-measure and $\beta$-measure of the
   reassemblings $(S_7,{\B}_1)$, $(S_7,{\B}_2)$, $(S_7,{\B}_3)$,
    and $(S_7,{\B}_4)$:
{  
\begin{alignat*}{5}
  & \alpha(S_7,{\B}_1) =
    \alpha(S_7,{\B}_2) =
    \alpha(S_7,{\B}_3) = \alpha(S_7,{\B}_4) = 7,
\\
  & \beta(S_7,{\B}_1) =\ 32,
   \quad \beta(S_7,{\B}_2) =\ 34,
\\
  & \beta(S_7,{\B}_3) =\ 35, \quad \beta(S_7,{\B}_4) =\ 31.
\end{alignat*}
}  
\!\!The preceding four reassemblings are all $\alpha$-optimal,
each for its own class of reassemblings, \ie, $(S_7,{\B}_2)$
is $\alpha$-optimal for the class of \emph{balanced} reassemblings of $S_7$
and $(S_7,{\B}_3)$ for the class of \emph{linear} reassemblings of $S_7$. 
It turns out that only $(S_7,{\B}_2)$  is $\beta$-optimal for its own 
class, the class of \emph{balanced} reassemblings of $S_7$. None of the four
is $\beta$-optimal for the class of all binary reassemblings of $S_7$.
\end{minipage}
&
\begin{minipage}{0.33\textwidth}
   \begin{figure}[H] 
     \begin{center}
       \includegraphics[scale=.35,trim=0cm 1.90cm 0cm 0cm,clip]
       {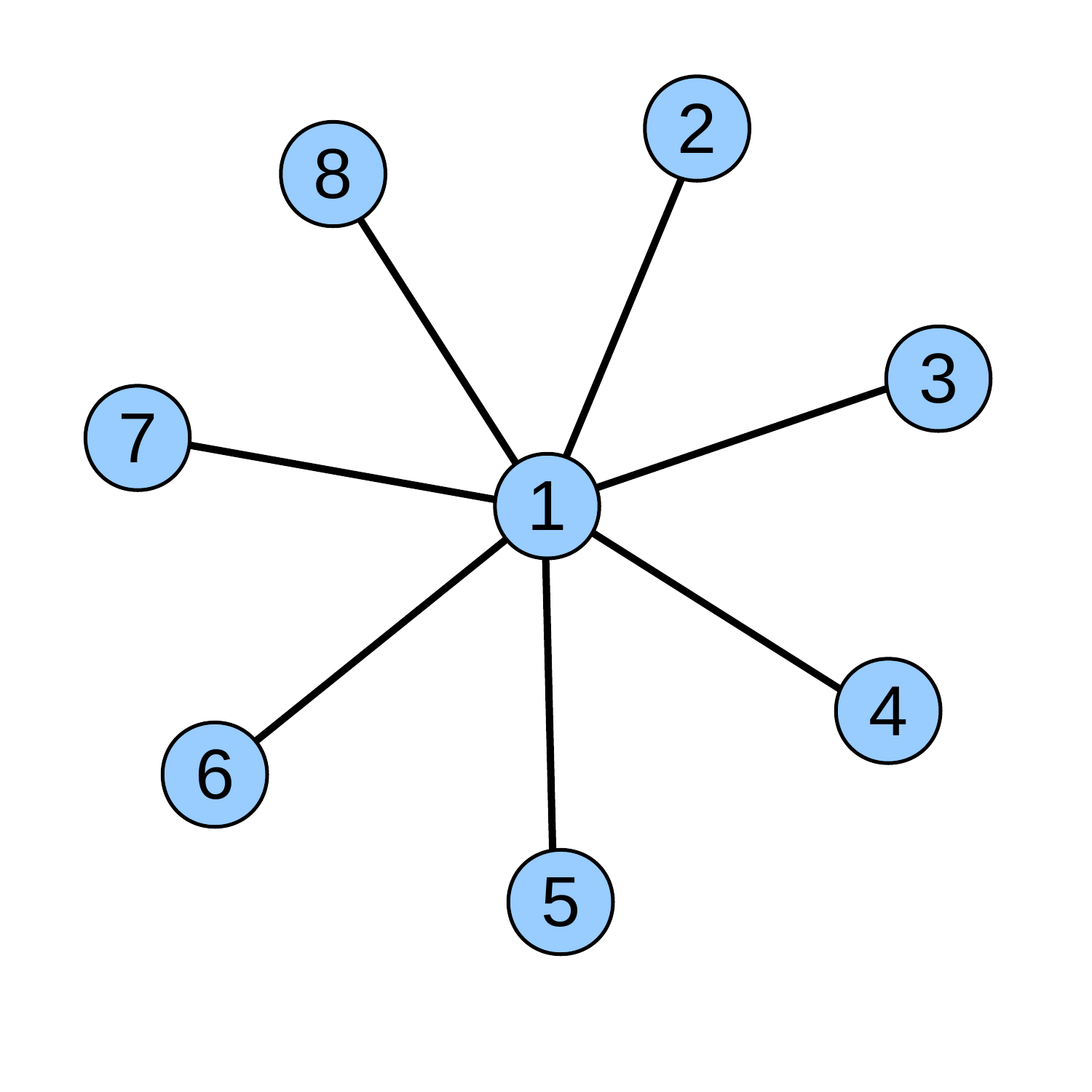}
     \end{center}
      \caption{\footnotesize Star graph $S_7$.}
      \label{fig:star}
   \end{figure}
\end{minipage}
\end{tabular}
\end{custommargins}

\smallskip
Of the four binary reassemblings above, only $(S_7,{\B}_3)$ is \emph{strict};
the three other are not strict, \ie, the three other merge some cluster
pairs $(A,B)$ such that $\bridges{}{A,B} = \varnothing$.
Because $(S_7,{\B}_1)$, $(S_7,{\B}_2)$, and $(S_7,{\B}_4)$ are
not strict, it is not possible to re-define them as
sequential reassemblings, each relative to an appropriate edge ordering.
An ordering ${\EdgePerm}_3^{S_7}$ of the edges that induces a sequential
reassembling equal to $(S_7,{\B}_3)$ is:
\[
  {\EdgePerm}_3^{S_7} = 
  \set{\ 1\ 2\ }\ \ \set{\ 1\ 3\ }\ \ \set{\ 1\ 4\ }\ \ \set{\ 1\ 5\ }
     \ \ \set{\ 1\ 6\ }\ \ \set{\ 1\ 7\ }\ \ \set{\ 1\ 8\ } .
\]
By exhaustive inspection (details omitted), the following
is a binary tree ${\B}_5$ over the vertices $\Set{1,\ldots,8}$
such that $(S_7,{\B}_5)$ is $\beta$-optimal for the class of all
binary reassemblings of $S_7$. Since $(S_7,{\B}_5)$ is not strict,
there is no corresponding ordering ${\EdgePerm}_5^{S_7}$ of the edges:
\[
  {\B}_5\ =\quad  
  \begin{cases}
  \quad 
  \set{\ \set{\ \set{\ \set{
  \ \set{\ \set{\ \set{\ 2\ 3\ }\ \ 4\ }\ \ 1\ }\ \ 5\ }\ \ 6\ }\ \ 7\ }\ \ 8\ }
  \qquad
  \end{cases}
\]
The resulting $\beta$-measure is $\beta(S_7,{\B}_5) = 29$. Note that 
${\B}_5$ is also linear. Hence, $(S_7,{\B}_5)$ is also
$\beta$-optimal for the class of linear reassemblings of $S_7$.
\end{example}

\begin{example}
\label{ex:from-arrangement-to-reassembling}
\label{ex:from-arrangement-to-reassembling-1}
The binary tree ${\B}_3$ in
Examples~\ref{ex:hypercube-graph}, \ref{ex:complete},
and~\ref{ex:star}, is a linear binary tree. Written in full, using the
notation in~\eqref{eq:linear1} at the beginning of
Section~\ref{sect:linear-reassembling}, the non-singleton sets of
${\B}_3$ are:
\begin{alignat*}{13}
  & X_1\ &&=\quad && 
  \set{\ 1\ 2\ }, \qquad
    &&  X_2\ &&=\quad && 
  \set{\ 1\ 2\ 3\ }, \qquad
    && X_3\ &&=\quad && 
  \set{\ 1\ 2\ 3\ 4\ }, \qquad
    && X_4\ &&=\quad && 
  \set{\ 1\ 2\ 3\ 4\ 5\ }, \qquad
\\[1ex]
  & X_5\ &&=\quad && 
  \set{\ 1 \ 2\ 3\ 4\ 5\ 6\ }, \qquad
  && X_6\ &&=\quad && 
  \set{\ 1\ 2\ 3\ 4\ 5\ 6\ 7\ }, \qquad
  && X_7\ &&=\quad && 
  \set{\ 1\ 2\ 3\ 4\ 5\ 6\ 7\ 8\ }.
  \qquad && 
\end{alignat*}
The singleton sets of ${\B}_3$ are:
\[
   Y_0 = \set{1},\quad Y_1 = \set{2},\quad Y_2 = \set{3},\quad Y_3 = \set{4},
  \quad Y_4 = \set{5},\quad Y_5 = \set{6},\quad Y_6 = \set{7},
   \quad Y_7 = \set{8} .
\]
The linear arrangement $\varphi_3$ induced by the linear reassembling
$(S_7,{\B}_3)$ in Example~\ref{ex:star} is 
(see Definition~\ref{def:linear-arrangement-induced}):
\[
   \varphi_3(2) = 1,\quad \varphi_3(1) = 2,\quad \varphi_3(3) = 3,
   \quad \varphi_3(4) = 4,\quad \varphi_3(5) = 5,\quad \varphi_3(6) = 6,
   \quad \varphi_3(7) = 7,\quad \varphi_3(8) = 8,
\]
rather than the linear arrangement $\varphi'$:
\[
   \varphi_3'(1) = 1,\quad \varphi_3'(2) = 2,\quad \varphi_3'(3) = 3,
   \quad \varphi_3'(4) = 4,\quad \varphi_3'(5) = 5,\quad \varphi_3'(6) = 6,
   \quad \varphi_3'(7) = 7,\quad \varphi_3'(8) = 8,
\]
because $\degr{S_7}{2} = 1 < 7 = \degr{S_7}{1}$. The difference between
$\varphi_3$ and $\varphi_3'$ is in the placement of the two first vertices:
vertex ``$1$'' and vertex ``$2$''.
\end{example}

\begin{example}
\label{ex:from-arrangement-to-reassembling-2}
The binary tree ${\B}_5$ in
Example~\ref{ex:star} is a linear binary tree. As in
Example~\ref{ex:from-arrangement-to-reassembling-1}, it is
straightforward to specify the singleton and non-singleton sets of
${\B}_5$ (omitted here) to fit the notation of~\eqref{eq:linear1} 
and~\eqref{eq:linear2}  at the beginning of
Section~\ref{sect:linear-reassembling}. There are two possible
linear arrangements, $\varphi_5$ and $\varphi_5'$,
which are induced by the linear reassembling $(S_7,{\B}_5)$,
because $\degr{S_7}{2} = \degr{S_7}{3} = 1$ 
(see Definition~\ref{def:linear-arrangement-induced}), namely:
\begin{alignat*}{12}
 & \varphi_5(2) = 1,\quad && \varphi_5(3) = 2,\quad && \varphi_5(4) = 3, \quad
 && \varphi_5(1) = 4,\quad && \varphi_5(5) = 5, \quad 
 && \varphi_5(6) = 6,\quad && \varphi_5(7) = 7,\quad && \varphi_5(8) = 8,
\\[1.5ex]
 & \varphi_5'(3) = 1,\quad && \varphi_5'(2) = 2,\quad && \varphi_5'(4) = 3, \quad
 && \varphi_5'(1) = 4,\quad && \varphi_5'(5) = 5, \quad 
 && \varphi_5'(6) = 6,\quad && \varphi_5'(7) = 7,\quad && \varphi_5'(8) = 8.
\end{alignat*}
The difference between
$\varphi_5$ and $\varphi_5'$ is in the placement of the two first vertices:
vertex ``$2$'' and vertex ``$3$''. In contrast to $\varphi_3$ and $\varphi_3'$
in Example~\ref{ex:from-arrangement-to-reassembling-1}, 
both $\varphi_5$ and $\varphi_5'$ are valid as linear arrangements
induced by the linear reassembling $(S_7,{\B}_5)$. A comparison
between $\varphi_3$ and $\varphi_5$ is shown in
Figure~\ref{fig:from-arrangement-to-reassembling}.
\end{example}

\begin{figure}[hc!]
  \begin{center}
  \begin{minipage}{0.45\textwidth}
     \includegraphics[scale=.6,trim=0cm 4.0cm 0cm 0cm,clip]%
          {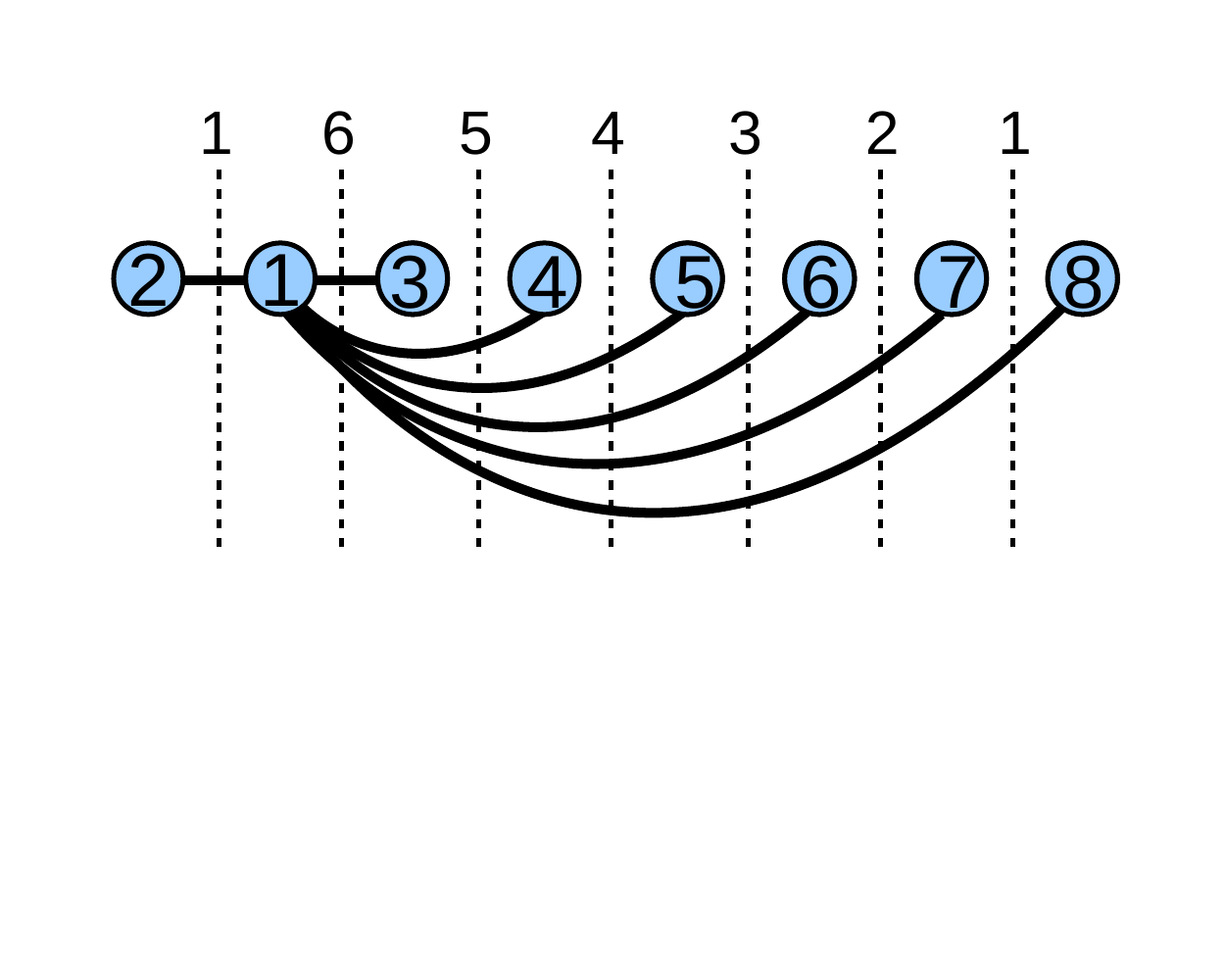}
     \begin{alignat*}{8}
     & \alpha(S_7,\varphi_3) =
       \ &&\max\;&& \Set{1,\,6,\,5,\,4,\,3,\,2,\,1}\;&&=\;&& 6
     \\
     & \beta(S_7,\varphi_3) =
       \ &&\sum\;&& \Set{1,\,6,\,5,\,4,\,3,\,2,\,1}\ &&=\;&& 22
     \end{alignat*}
  \end{minipage}
     \quad 
  \begin{minipage}{0.45\textwidth}
     \includegraphics[scale=.6,trim=0cm 4.0cm 0cm 0cm,clip]
          {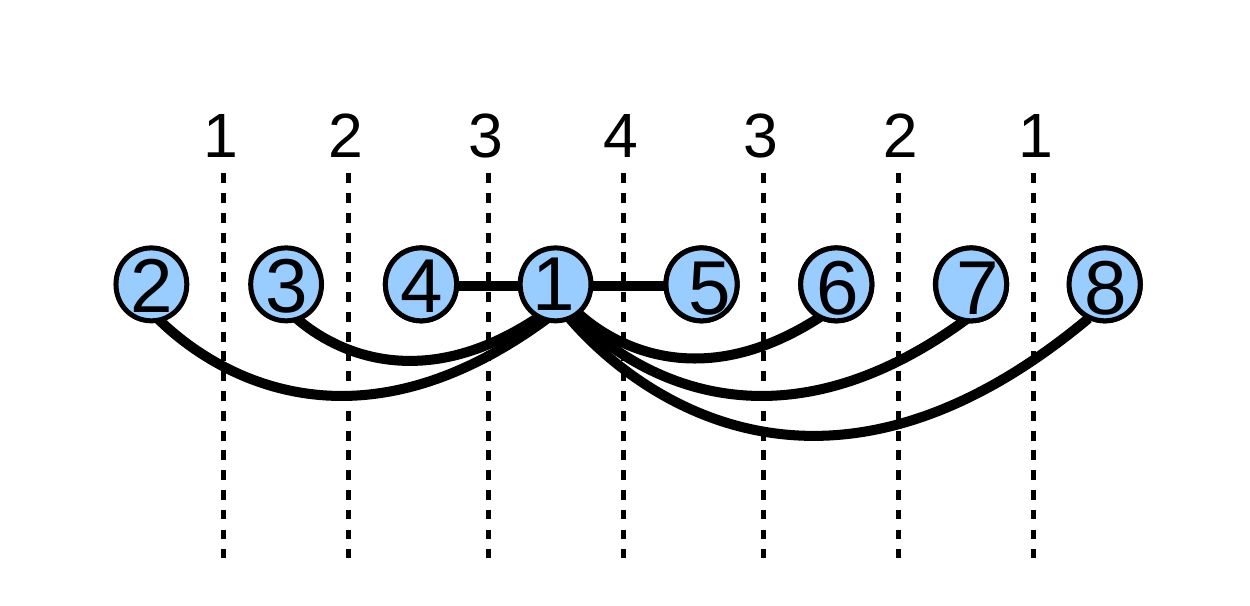}
     \begin{alignat*}{8}
     & \alpha(S_7,\varphi_5) =
       \ &&\max\;&& \Set{1,\,2,\,3,\,4,\,3,\,2,\,1}\;&&=\;&& 4
     \\
     & \beta(S_7,\varphi_5) =
       \ &&\sum\;&& \Set{1,\,2,\,3,\,4,\,3,\,2,\,1}\ &&=\;&& 16
     \end{alignat*}
  \end{minipage}
  \end{center}
     \caption{Comparison of linear arrangements $(S_7,\varphi_3)$ in
     Example~\ref{ex:from-arrangement-to-reassembling-1} and $(S_7,\varphi_5)$ 
     in Example~\ref{ex:from-arrangement-to-reassembling-2}.}
   \label{fig:from-arrangement-to-reassembling}
\end{figure}

\begin{example}
\label{ex:from-reassembling-to-arrangement}
This is a continuation of Example~\ref{ex:from-arrangement-to-reassembling}.
The linear reassembling induced by the linear arrangement $(S_7,\varphi_3)$ 
is precisely $(S_7,{\B}_3)$, but so is the linear reassembling induced by 
the linear arrangement $(S_7,\varphi_3')$ again the same $(S_7,{\B}_3)$,
according to Definition~\ref{def:linear-reassembling-induced}. 
This means: \emph{linear arrangements make distinction that linear
reassemblings do not make.} This difference is in the placement
of the two first vertices, specifically:
\begin{itemize}[itemsep=0pt,parsep=2pt,topsep=5pt,partopsep=0pt] 
\item The $\alpha$-measure and $\beta$-measure of a \emph{linear arrangement}
      generally depend on which vertex is placed first and
      which is placed second.
\item The $\alpha$-measure and $\beta$-measure of a \emph{linear reassembling}
      do \emph{not} distinguish between a first and second vertex
      and do \emph{not} depend on which is placed first and
      which is placed second.
\end{itemize}
As an example, consider the linear arrangements  $(S_7,\varphi_3)$ 
and  $(S_7,\varphi_3')$. Their $\alpha$-measure and $\beta$-measure are:
\begin{alignat*}{7}
    & \alpha(S_7,\varphi_3) \ &&=\ && 6,\qquad && \beta(S_7,\varphi_3) 
    \ &&=\ && 22 ,
\\
    & \alpha(S_7,\varphi_3') \ &&=\ && 7,\qquad && \beta(S_7,\varphi_3') 
    \ &&=\ && 28 .
\end{alignat*}
For the linear reassembling $(S_7,{\B}_3)$
induced by both $(S_7,\varphi_3)$ and  $(S_7,\varphi_3')$, we have:
$\alpha(S_7,{\B}_3) = 7$ and $\beta(S_7,{\B}_3) = 35$, as noted
in Example~\ref{ex:star}. Moreover, while both $(S_7,\varphi_3)$ and
$(S_7,\varphi_3')$ are neither $\alpha$-optimal nor $\beta$-optimal, 
$(S_7,{\B}_3)$ is $\alpha$-optimal (though not $\beta$-optimal).
\end{example}

\section{$\alpha$-Optimization of Linear Reassembling Is NP-Hard}
\label{sect:alpha-optimization}

We prove that the \emph{$\alpha$-optimality of linear arrangements} 
(in the literature: the \emph{minimum-cutwidth linear arrangement} problem)
and the \emph{$\alpha$-optimality of linear reassemblings} are reducible to
each other in polynomial time.

\begin{definition}{Chordal graph, triangulation, clique number, treewidth}
\label{def:chordal}
Let $G = (V,E)$ be a simple undirected graph. The following are standard
notions of graph theory~\cite{diestel2012}.
\begin{itemize}[itemsep=0pt,parsep=2pt,topsep=5pt,partopsep=0pt] 
\item
   $G$ is a \emph{chordal graph} if every cycle of length of $4$ or more has 
   a \emph{chord}, \ie, an edge connecting two vertices that are not consecutive
   in the cycle.
\item
   A \emph{triangulation of $G$} is a chordal graph $G' = (V',E')$ where $V = V'$ 
   and $E\subseteq E'$. In such a case, we say that $G$ can be 
   \emph{triangulated into} $G'$, not uniquely in general.
\item
   The \emph{clique number} of $G$, denoted $\omega(G)$, is the size of a
   largest clique in $G$.
\item
   There are different equivalent definitions of
   the \emph{treewidth}. We here use a definition, or a consequence of
   the original definition, which is more convenient for our
   purposes~\cite{BronnerRies2006,diestel2012}.  The \emph{treewidth}
   of $G$ is:
\[
   \min\,\Set{\,\omega(G')\;|\;\text{$G'$ is a triangulation of $G$}\,}\ -\ 1.
\]
\end{itemize}
In words, among all triangulations $G'$ of $G$, we choose a $G'$ whose
clique number is smallest: The \emph{treewidth} of $G$ is one less than
the clique number of such a $G'$.
\end{definition}

\begin{lemma}
\label{lem:thilikos1}
For every positive integers $\Delta$ and $k$,
there is an algorithm $\A$ using $\Delta$ and $k$ as fixed parameters,
such that, given an arbitrary simple undirected graph
$G = (V,E)$ as input to $\A$, if:
\begin{enumerate}[itemsep=0pt,parsep=2pt,topsep=5pt,partopsep=0pt] 
\item the maximum vertex degree in $G$ is $\leqslant \Delta$, and
\item the treewidth of $G$ is $\leqslant k$,
\end{enumerate}
then $\A$ computes
a minimum-cutwidth linear arrangement of $G$ in time $\bigO{n^{\Delta k^2}}$
where  $n = \size{V}$.
\end{lemma}

\begin{proof}
This is Theorem 4.2 in~\cite{thilikosSernaBodlaender2001}, where
the algorithm not only computes the value of a minimum cutwidth, but
can be adjusted to output the corresponding  minimum-cutwidth
linear arrangement $\varphi : V\to\Set{1,\ldots,n}$.
\end{proof}

In Lemma~\ref{lem:thilikos2}, 
a $\emph{cut vertex}$ in $G$ is a vertex whose removal
increases the number of connected components.

\begin{lemma}
\label{lem:thilikos2}
There is an algorithm $\A$ such that, given an arbitrary simple undirected graph
$G = (V,E)$ as input to $\A$, if:
\begin{enumerate}[itemsep=0pt,parsep=2pt,topsep=5pt,partopsep=0pt] 
\item every vertex in $G$ has degree $\leqslant 3$, and
\item every vertex in $G$ of degree $= 3$ is a cut vertex,
\end{enumerate}
then $\A$ computes 
a minimum-cutwidth linear arrangement of $G$ in time $\bigO{n^{12}}$
where  $n = \size{V}$.
\end{lemma}

\begin{proof}
We show that the treewidth $k$ of $G$ is $\leqslant 2$. Because 
the maximum vertex degree $\Delta$ of $G$ is $\leqslant 3$, 
Lemma~\ref{lem:thilikos1} implies the existence of an algorithm $\A$ which
runs in time  $\bigO{n^{\Delta k^2}} = \bigO{n^{12}}$.

To show that $k\leqslant 2$,
consider a vertex $v$ of degree $= 3$, which is therefore a cut vertex.
The removal of $v$ can have one of two possible outcomes: 
\begin{enumerate}[itemsep=0pt,parsep=2pt,topsep=5pt,partopsep=0pt] 
\item[(a)] disconnect $G$ into $3$ components, or
\item[(b)] disconnect $G$ into $2$ components.
\end{enumerate}
If every vertex of degree $= 3$ satisfies condition (a), then $G$ is tree whose
treewidth is $1$, since its clique number $\omega(G) = 2$ in this case.

If $C_1$ and $C_2$ are cycles in $G$, each with $3$ vertices or more,
then $C_1$ and $C_2$ are non-overlapping, \ie, $C_1$ and $C_2$ have no
vertex in common and no edge in common. If they have an edge
$\set{v_1\,w_1}$ in common, then there is an edge $\set{v_2\,w_2}\in
C_1\cap C_2$ such that $\degree{v_2}=3$ (or, resp., $\degree{w_2}=3$) and 
$v_2$ (or, resp., $w_2$) is not a cut vertex, contradicting the
hypothesis. If $C_1$ and $C_2$ have no edge in common, but do have
a vertex $v$ in common, then $\degree{v_2} > 3$, again contradicting
the hypothesis.

In case one or more vertices satisfy condition (b), $G$ can be
therefore viewed as a finite collection of non-overlapping
rings $\Set{R_1,\ldots,R_p}$, each ring being a cycle with at least
$3$ vertices, satisfying condition (c):
\begin{enumerate}[itemsep=0pt,parsep=2pt,topsep=5pt,partopsep=0pt] 
\item[(c)] if two distinct rings $\Set{R_i,R_j}$, with $i\neq j$, 
     are connected by a path $P_{i,j}$, then the removal of all the
     vertices and edges of $P_{i,j}$
     (in particular the two endpoints of $P_{i,j}$, one in $R_i$ and
     one in $R_j$, which are necessarily vetices of degree $=3$)
     disconnects $G$ into $2$ components.
\end{enumerate}
Another way of expressing (c) is that, if all the rings 
$\Set{R_1,\ldots,R_p}$ are each contracted to a single vertex, then the
result is a tree (where some of the internal vertices may now have degree 
larger than $3$). Since the clique number of a ring is $3$,
the treewidth of a ring is $2$, and the desired
conclusion follows.
\end{proof}

\begin{lemma}
\label{lem:one-vertex-not-a-cut-vertex1}
Let $G = (V,E)$ be a simple undirected graph, where
every vertex has degree $\leqslant 3$, and let $(G,\LL)$
be a linear reassembling of $G$. Consider the longest
chain of nested clusters of size $\geqslant 2$, as in~\eqref{eq:linear1}
in the opening paragraph of Section~\ref{sect:linear-reassembling}: 
\[
   X_1 \ \subsetneq X_2 
                  \ \subsetneq\ \cdots\ \subsetneq\ X_{n-1} = V .
\]
\textbf{Conclusion:} If there is one vertex of degree $= 3$ in $G$
which is \emph{not} a cut vertex, then \\
$\max\,\Set{\,\degr{}{X_1},\ldots,\degr{}{X_{n-1}}}\,\geqslant 3$.
\end{lemma}

\begin{proof}
We first show there are least two vertices of degree $=3$ which are
not cut vertices. Let $v$ be a vertex of degree $= 3$ which is not a
cut vertex, and let $\Set{\set{v\,x},\set{v\,y},\set{v\,z}}$ be the
three edges incident to $v$. 
Because $v$ is not a cut vertex, any two edges in 
$\Set{\set{v\,x},\set{v\,y},\set{v\,z}}$ are consecutive edges in
a cycle containing $v$. Let $C(v,x,y)$ be a cycle containing
edges $\Set{\set{v\,x},\set{v\,y}}$, and define similarly 
cycles $C(v,x,z)$ and $C(v,y,z)$. If any of these three cycles contains
a chord, then the two endpoints of the chord are vertices of degree $=3$
which are not cut vertices. If none of these three cycles contain a chord,
then we can combine any two of them, because they share an edge, to 
form another cycle with a chord, which again implies the existence of
two vertices of degree $=3$ which are not cut vertices.

To conclude the proof, consider the clusters of $\LL$ 
of size $\geqslant 2$: These are $\Set{X_1,\ldots,X_{n-1}}$, and
the corresponding singleton clusters are $\Set{Y_0,\ldots,Y_{n-1}}$, 
as in~\eqref{eq:linear1} and~\eqref{eq:linear2} 
in Section~\ref{sect:linear-reassembling}. By the preceding
argument, there are at least two vertices of degree $=3$ which are 
not cut vertices. Let one of these two be $v$, with $Y_i = \Set{v}$ 
for some $i\geqslant 1$. 

We have $X_{i-1} \cap Y_i =\varnothing$ and $X_{i-1} \cup Y_i = X_i$.
There are two cases: (1) For some vertex $w\in X_{i-1}$, there is an
edge $\set{v\,w}\in E$, and (2) for every vertex $w\in X_{i-1}$, there
is no such edge. We consider case (1) and leave the other (easier)
case (2) to the reader.

We cannot have $\degr{}{X_{i-1}} = 0$, otherwise $G$ is
disconnected, nor can we have $\degr{}{X_{i-1}} = 1$, otherwise
$v$ is a cut vertex. Hence, $\degr{}{X_{i-1}} \geqslant 2$. If
$\degr{}{X_{i-1}} \geqslant 3$, this is already the conclusion of
the lemma and there is nothing else to prove.
Suppose $\degr{}{X_{i-1}} = 2$, the case left to consider.

Similarly, we cannot have $\degr{}{X_{i}} = 0$, otherwise $G$ is
disconnected, nor $\degr{}{X_{i}} = 1$, otherwise
$v$ is a cut vertex. Hence, $\degr{}{X_{i}} \geqslant 2$.
But if $\degr{}{X_{i-1}} = 2$ and $\degr{}{X_{i}} \geqslant 2$,
with $\degree{v} = 3$, then it must be that 
$\degr{}{X_{i}} = 3$.
\end{proof}

\begin{lemma}
\label{lem:one-vertex-not-a-cut-vertex2}
Let $G = (V,E)$ be a simple undirected graph, where
every vertex has degree $\leqslant 3$ and where one vertex 
of degree $= 3$ is \emph{not} a cut vertex.
Let $(G,\LL)$ be a linear reassembling and $(G,\varphi)$ a linear
arrangement. 

\noindent
\textbf{Conclusion:} If $(G,\LL)$ is induced by $(G,\varphi)$,
or if $(G,\varphi)$ is induced by $(G,\LL)$, then:
\begin{itemize}[itemsep=0pt,parsep=2pt,topsep=5pt,partopsep=0pt] 
    \item $\alpha(G,\LL) = \alpha(G,\varphi)$.
    \item $(G,\LL)$ is $\alpha$-optimal iff $(G,\varphi)$ is $\alpha$-optimal.
\end{itemize}
\end{lemma}

\begin{proof}
Straightforward consequence of Lemma~\ref{lem:one-vertex-not-a-cut-vertex1},
the definitions of $\alpha(G,\LL)$ and $\alpha(G,\varphi)$, and what it
means for $(G,\LL)$ to be induced by $(G,\varphi)$ and 
for $(G,\varphi)$ to be induced by $(G,\LL)$. When there is  
at least one vertex of degree $= 3$ which is \emph{not} a cut vertex,
and therefore at least two of them by the proof of 
Lemma~\ref{lem:one-vertex-not-a-cut-vertex1}, we can
ignore the degrees of singleton clusters in the computation of 
$\alpha(G,\LL)$. All details omitted.
\end{proof}

\begin{theorem}
\label{thm:reducing-linear-arrangements-to-linear-reassemblings}
For the class of simple undirected graphs $G$ where every vertex
has degree $\leqslant 3$, the $\alpha$-optimality of linear
arrangements $(G,\varphi)$ is polynomial-time reducible to the
$\alpha$-optimality of linear reassemblings $(G,\LL)$.
\end{theorem}

More explicitly, a polynomial-time algorithm $\A$, which 
returns an $\alpha$-optimal linear reassembling $(G,\LL)$ of
a graph $G$ where every vertex has degree $\leqslant 3$,
can be used to return an $\alpha$-optimal linear arrangement $(G,\varphi)$.

\begin{proof}
Consider an arbitrary $G$ where every vertex has degree $\leqslant 3$.
If every vertex in $G$ of degree $= 3$ is a cut vertex, we use the
algorithm in Lemma~\ref{lem:thilikos2} to compute an $\alpha$-optimal
linear arrangement $(G,\varphi)$ in time $\bigOO{n^{12}}$. If
there is a vertex in $G$ of degree $= 3$ which is not a cut vertex,
we first compute an $\alpha$-optimal linear reassembling $(G,\LL)$
and then return the linear arrangement $(G,\varphi)$ induced by
$(G,\LL)$. The desired conclusion follows from 
Lemma~\ref{lem:one-vertex-not-a-cut-vertex2}.
\end{proof}

\begin{corollary}
\label{cor:NP-hardness-of-alpha-optimality-of-linear-reassemblings}
For the class of all simple undirected graphs $G$,
the computation of $\alpha$-optimal linear reassemblings $(G,\LL)$
is an NP-hard problem.
\end{corollary}

\begin{proof}
If there is a polynomial-time algorithm $\A$ to compute, for an 
arbitrary simple undirected graph, an $\alpha$-optimal
linear reassembling, then the same algorithm $\A$
can be used to compute in polynomial-time an $\alpha$-optimal
linear reassembling $(G,\LL)$ for a graph $G$  where every vertex
has degree $\leqslant 3$. By 
Theorem~\ref{thm:reducing-linear-arrangements-to-linear-reassemblings},
$\A$ can be further adapted to compute an $\alpha$-optimal
linear arrangement $(G,\varphi)$ for such a graph $G$ in polynomial time.
But the latter problem (in the literature: 
the \emph{minimum-cutwidth linear arrangement} problem) is known 
to be NP-hard~\cite{makedon1985,monien1988}.
\end{proof}

\begin{remark}
\label{rem:open-problem1}
To the best of our knowledge, the complexity status of
the \emph{minimum-cutwidth linear
arrangement} problem for $k$-regular graphs for a fixed $k\geqslant 3$
is an open problem. If it were known to be NP-hard, we would be able
to simplify our proof of
Theorem~\ref{thm:reducing-linear-arrangements-to-linear-reassemblings}
and its corollary considerably. In particular, we would be able to
eliminate Lemmas~\ref{lem:thilikos1} and~\ref{lem:thilikos2} and the
supporting Definition~\ref{def:chordal}, as well as simplify
Lemmas~\ref{lem:one-vertex-not-a-cut-vertex1}
and~\ref{lem:one-vertex-not-a-cut-vertex2} by restricting them to
$k$-regular graphs.
\end{remark}

Theorem~\ref{thm:reducing-linear-arrangements-to-linear-reassemblings} and
Corollary~\ref{cor:NP-hardness-of-alpha-optimality-of-linear-reassemblings}
together say the $\alpha$-optimality of linear arrangements $(G,\varphi)$
is polynomial-time reducible to the $\alpha$-optimality of linear
reassemblings $(G,\LL)$.
For completeness, we show the converse in the next theorem.

\begin{theorem}
\label{thm:reducing-linear-reassemblings-to-linear-arrangements}
For the class of simple undirected graphs $G$ in general,
the $\alpha$-optimality of linear reassemblings $(G,\LL)$ is 
polynomial-time reducible to the $\alpha$-optimality of linear
arrangements $(G,\varphi)$.
\end{theorem}

\begin{proof}
In Appendix~\ref{sect:remaining-proofs-for-linear}.
\end{proof}

\section{$\beta$-Optimization of Linear Reassembling Is NP-Hard}
\label{sect:beta-optimization}

We prove that the \emph{$\beta$-optimality of linear arrangements} 
(in the literature: the \emph{minimum-cost linear arrangement} problem
or also the \emph{optimal linear arrangement} problem)
and the \emph{$\beta$-optimality of linear reassemblings} are reducible to
each other in polynomial time. Towards this end, we prove an 
intermediate result, which is also of independent interest 
(Theorem~\ref{thm:equivalence-of-anchored-beta-quivalence} which
presupposes Definition~\ref{def:anchored}).

\begin{definition}{Anchored linear reassemblings}
\label{def:restricted}
\label{def:anchored}
Let $G = (V,E)$ be a simple undirected graph and let $w\in V$.
Let $(G,\LL)$ be a linear reassembling of $G$, whose longest
chain of nested clusters of size $\geqslant 2$, as in~\eqref{eq:linear1}
in the opening paragraph of Section~\ref{sect:linear-reassembling}, is: 
\[
   X_1 \ \subsetneq X_2 
                  \ \subsetneq\ \cdots\ \subsetneq\ X_{n-1} = V 
\]
and whose corresponding singleton clusters are $\Set{Y_0,\ldots,Y_{n-1}}$,
as determined by~\eqref{eq:linear2} in the opening paragraph of 
Section~\ref{sect:linear-reassembling}. We say $(G,\LL)$ is a 
\emph{linear reassembling anchored at $w\in V$} iff
there is a vertex $w'\in V$ such that:
\[
    Y_0 = \Set{w},\quad Y_1 = \Set{w'},\quad\text{and}\  
    \degr{G}{w}\leqslant\degr{G}{w'}.
\]
Note that we require that the immediate sibling $Y_1 = \Set{w'}$ of
the leaf node $Y_0 = \Set{w}$ satisfy the condition
$\degr{G}{w}\leqslant\degr{G}{w'}$. This implies that, given an
arbitrary vertex $w\in V$, we cannot anchor a linear reassembling at
$w$ unless we find another vertex $w'\in V$ such that
$\degr{G}{w}\leqslant\degr{G}{w'}$ and then make $\Set{w}$ and
$\Set{w'}$ sibling leaf-nodes. This is a technical restriction to
simplify the statement of
Lemma~\ref{lem:equivalence-when-starting-vertex-is-the-same}.%
   \footnote{Thus, we cannot say that the linear reassembling
    $(S_7,{\B}_3)$, in Examples~\ref{ex:star}
    and~\ref{ex:from-arrangement-to-reassembling}, is anchored at
    vertex ``$1$'', though we can say that $(S_7,{\B}_3)$ is anchored
    at vertex ``$2$''.  More generally, in the case of a star graph
    $S_k$ with $k\geqslant 3$ leaves: There is no linear reassembling
    of $S_k$ anchored at the internal vertex of $S_k$.}
\end{definition}

\begin{definition}{Anchored linear arrangements}
\label{def:anchored-arrangement}
Let $G = (V,E)$ be a simple undirected graph and let $w\in V$.
Let $(G,\varphi)$ be a linear arrangement of $G$. We say $(G,\varphi)$ is a 
\emph{linear arrangement anchored at $w\in V$} iff there is 
a vertex $w'\in V$ such that: 
\[
    \varphi(w) = 1,\quad \varphi(w') = 2,\quad\text{and}\  
    \degr{G}{w}\leqslant\degr{G}{w'}.
\]
Again, as in Definition~\ref{def:anchored}, the condition
$\degr{G}{w}\leqslant\degr{G}{w'}$ is imposed in order to simplify the
statement of Lemma~\ref{lem:equivalence-when-starting-vertex-is-the-same}.  
It is worth noting that, if we relax this condition and allow
$\degr{G}{w} > \degr{G}{w'}$, then the new arrangement $\varphi'$
which permutes the positions of $w$ and $w'$, \ie:
\[
   \varphi'(w') = 1, \quad \varphi'(w) = 2, \quad\text{and}
   \ \varphi'(v) = \varphi(v)\ \text{for all $v\in V-\Set{w,w'}$},
\]
is such that 
$\beta(G,\varphi') < \beta(G,\varphi)$. In words, if we allowed
$\degr{G}{w} > \degr{G}{w'}$, the linear arrangement $(G,\varphi)$ would
not be $\beta$-optimal.%
   \footnote{A similar statement applies to the 
   $\alpha$-measure: If we allowed $\degr{G}{w} > \degr{G}{w'}$, 
   then the new arrangement $\varphi'$ would be such that
   $\alpha(G,\varphi') \leqslant \alpha(G,\varphi)$, but
   not necessarily $\alpha(G,\varphi') < \alpha(G,\varphi)$.
   }
\end{definition}

\begin{example}
\label{ex:anchored}
Consider the linear reassemblings $(S_7,{\B}_3)$ and $(S_7,{\B}_5)$ in
Example~\ref{ex:star}.  $(S_7,{\B}_3)$ is anchored at vertex ``$2$'',
but cannot be anchored at vertex ``$1$'', while $(S_7,{\B}_5)$ is
anchored at vertex ``$2$'', and can be anchored again at vertex
``$3$''.
Both $(S_7,{\B}_3)$ and $(S_7,{\B}_5)$ are $\alpha$-optimal and,
a fortiori, $\alpha$-optimal for the class of all linear reassemblings
of $S_7$ anchored at vertex ``$2$''.
Moreover, $\beta(S_7,{\B}_3) = 35$ and $\beta(S_7,{\B}_5) = 29$, so
that $(S_7,{\B}_3)$ is not $\beta$-optimal for the class of all linear
reassemblings of $S_7$ anchored at vertex ``$2$'', while
$(S_7,{\B}_5)$ is $\beta$-optimal for the same class.

Consider now the linear arrangements $(S_7,\varphi_3)$ and 
$(S_7,\varphi_5)$ induced by the linear
reassemblings $(S_7,{\B}_3)$ and $(S_7,{\B}_5)$, respectively.
$\varphi_3$ and $\varphi_5$ are given in  
Example~\ref{ex:from-arrangement-to-reassembling-1}
and Example~\ref{ex:from-arrangement-to-reassembling-2}.
Both $(S_7,\varphi_3)$ and $(S_7,\varphi_5)$ are anchored at vertex
``$2$''. Moreover, $(S_7,\varphi_3)$ cannot be anchored at vertex
``$1$'' (the sibling leaf of ``$2$'' in $\varphi_3$), while $(S_7,\varphi_5)$
can be anchored again at vertex ``$3$'' (the sibling leaf of ``$2$'' in 
$\varphi_5$).

$(S_7,\varphi_3)$ is neither $\alpha$-optimal nor $\beta$-optimal
for the class of all linear arrangements of $S_7$ anchored at
``$2$''; hence, $(S_7,\varphi_3)$ is neither $\alpha$-optimal nor 
$\beta$-optimal for the super-class of all linear arrangements of $S_7$.
By contrast, $(S_7,\varphi_5)$ is both $\alpha$-optimal and $\beta$-optimal
for the class of all linear arrangements of $S_7$;
hence, $(S_7,\varphi_5)$ is both $\alpha$-optimal and 
$\beta$-optimal for the sub-class of all linear arrangements of $S_7$
anchored at ``$2$''.
\end{example}

Let $(G,{\LL})$ be a linear reassembling anchored at vertex $w\in V$.
We say $(G,{\LL})$ is \emph{$\beta$-optimal relative to anchor $w$} iff: 
\[
   \beta(G,{\LL})\ =
   \ \min\,\SET{\,\beta(G,{\LL'})\;\bigl|
   \;(G,{\LL'})\text{ is a linear reassembling anchored at $w$}\,} .
\]
Clearly, $(G,{\LL})$ is a \emph{$\beta$-optimal 
linear reassembling}, with no anchor restriction, iff:
\begin{alignat*}{5}
   & \beta(G,{\LL})\ &&=
   \ &&\min\,\bigl\{\,\beta(G,{\LL'})\;\bigl|
   \;&&\text{there is a vertex $w\in V$ and}
   \\
   & && && &&(G,{\LL'})\text{ is a linear reassembling $\beta$-optimal relative 
       to anchor $w$}\,\bigr\}.
\end{alignat*}
Similar obvious conditions apply to what it means
for $(G,{\varphi})$ to be a \emph{$\beta$-optimal 
linear arrangement relative to anchor $w$}.

\begin{lemma}
\label{lem:equivalence-when-starting-vertex-is-the-same}
Let $G = (V,E)$ be a simple undirected graph and $w\in V$.  Let
$(G,\LL)$ be a linear reassembling of $G$ anchored at $w$, and
$(G,\varphi)$ be a linear arrangement of $G$ anchored at $w$, such
that:
\begin{itemize}[itemsep=0pt,parsep=5pt,topsep=5pt,partopsep=0pt] 
\item[] $(G,\varphi)$ is induced by $(G,\LL)$ 
        \quad or\quad $(G,\LL)$ is induced by
        $(G,\varphi)$.
\end{itemize}
%
\textbf{Conclusion:} $(G,\LL)$ is
$\beta$-optimal relative to anchor $w$ iff $(G,\varphi)$ is
$\beta$-optimal relative to anchor $w$.
\end{lemma}

\begin{proof}
Let $d = \degr{}{w} \geqslant 1$ and 
$\Delta = \sum\Set{\degr{}{v}\,|\,v\in V\text{ and } v\neq w}$.
Consider the case when arrangement $(G,\varphi)$ is induced 
by reassembling $(G,\LL)$. (We omit the case when reassembling
$(G,\LL)$ is induced by arrangement $(G,\varphi)$, which is
treated similarly.) From Definition~\ref{def:different-measures}, 
\[
  \beta(G,\LL) = d + \Delta + 
  \sum\,\Set{\,\degree{X_i}\;|\;1\leqslant i\leqslant n-1\,}
\]  
where $X_1,\ldots,X_{n-1}$ are all the clusters of size $\geqslant 2$
in $\LL$. From Definitions~\ref{def:linear-arrangement}
and~\ref{def:linear-arrangement-induced}, 
\[
  \beta(G,\varphi) = d + 
  \sum\,\Set{\,\degree{X_i}\;|\;1\leqslant i\leqslant n-1\,} .
\]  
Hence, both $\beta(G,\LL)$ and $\beta(G,\varphi)$ are minimized when
the same quantity
$\sum\Set{\degree{X_i}\,|\,1\leqslant i\leqslant n-1}$ is minimized.
The desired conclusion follows.
\end{proof}

\begin{remark}
\label{rem:anchored}
There is an obvious definition of \emph{anchored $\alpha$-optimality},
similar to that of \emph{anchored $\beta$-optimality} above.
However, results for the latter do not necessarily 
have counterparts for the former.  In particular, the conclusion of
Lemma~\ref{lem:equivalence-when-starting-vertex-is-the-same} does not
hold for $\alpha$-optimality. Specifically, there are simple
counter-examples showing the existence of a simple undirected graph
$G(V,E)$ with a distinguished vertex $w\in V$ such that:
\begin{itemize}[itemsep=0pt,parsep=2pt,topsep=5pt,partopsep=0pt] 
\item there is a linear reassembling $(G,\LL)$ which is $\alpha$-optimal
      relative to anchor $w$,
\item but the linear arrangement $(G,\varphi)$ induced by  $(G,\LL)$
      is not $\alpha$-optimal relative to anchor $w$.
\end{itemize}
Such a counter-example is the linear reassembling $(S_7,{\B}_3)$
and the linear arrangement $(S_7,\varphi_3)$ it induces, in 
Example~\ref{ex:anchored}, both anchored at vertex ``$2$'':
the former is $\alpha$-optimal for the class of all linear reassemblings
of $S_7$ anchored at ``$2$'', the latter is not
$\alpha$-optimal for the class of all linear arrangements
of $S_7$ anchored at ``$2$''.

There is an examination, yet to be undertaken, of the relation between
linear reassemblings $(G,\LL)$ and linear arrangements $(G,\varphi)$
that are $\alpha$-optimal relative to the same anchor, similar to our
study of anchored $\beta$-optimality below. This examination we
do not pursue in this report.
\end{remark}

\begin{theorem}
\label{thm:equivalence-of-anchored-beta-quivalence}
For the class of all simple undirected graphs $G = (V,E)$,
each with a distinguished vertex $w\in V$, the two following
problems are polynomial-time reducible to each other:
\begin{itemize}[itemsep=0pt,parsep=2pt,topsep=5pt,partopsep=0pt] 
\item the $\beta$-optimality of linear arrangements $(G,\varphi)$
         anchored at $w$,
\item the $\beta$-optimality of linear reassemblings $(G,\LL)$
         anchored at $w$.
\end{itemize}
\end{theorem}

\noindent
More explicitly, a polynomial-time algorithm $\A$, which returns a
linear reassembling $(G,\LL)$ [resp. a linear arrangement
$(G,\varphi)$] which is $\beta$-optimal relative to anchor $w$ can be
used to return in polynomial time a linear arrangement $(G,\varphi)$
[resp. a linear reassembling $(G,\LL)$] which is $\beta$-optimal
relative to anchor $w$.

\begin{proof} 
This is an immediate consequence of 
Lemma~\ref{lem:equivalence-when-starting-vertex-is-the-same}.
\end{proof}

\begin{definition}{Auxiliary graphs}
\label{def:auxiliary-graphs}
Let $G = (V,E)$ be a simple undirected graph, with $\size{V} = n$
and $\size{E} = m$. For every $w\in V$ we define what we call an 
\emph{auxiliary graph} $G_w = (V_w,E_w)$ as follows:
\begin{itemize}[itemsep=0pt,parsep=5pt,topsep=5pt,partopsep=0pt] 
\item $V_w := V\uplus U$  where $U$
      is a fresh set of $p = \sum\,\Set{\,\degr{G}{v}\;|\;v\in V\,}$ vertices.
\item $E_w := E\uplus D_w$ where
      \(
      D_w := \Set{\,\set{u\,w}\;|\;u\in U}\;\cup\;
             \Set{\,\set{u_1\,u_2}\;|\;u_1,u_2\in U\text{ and } u_1\neq u_2\,}.
      \)
\end{itemize}
Thus, the subgraph of $G_w$ induced by the set $V$ is simply the
original graph $G$, and the subgraph of $G_w$ induced by the set 
$U\cup\Set{w}$ is the complete graph $K_{p+1}$ over $p+1$ vertices.

Informally, $G_w$ is constructed from $G$ and the complete graph
$K_{p+1}$ by identifying vertex $w\in V$ with one of the
vertices of $K_{p+1}$. In particular, $w$ is a cut vertex of the
auxiliary graph $G_w$. We call $w$, which is the common vertex
of $G$ and $K_{p+1}$, the \emph{distinguished vertex} of $G_w$.
\end{definition}

\begin{lemma}
\label{lem:size-of-auxiliary}
If  $G_w = (V_w,E_w)$ is the auxiliary graph for vertex $w\in V$,
as constructed in Definition~\ref{def:auxiliary-graphs},
then $\size{V_w}\leqslant n^2$ and
$\size{E_w}\leqslant (n^4-2n^3+3n^2-2n)/2$.
\end{lemma}

\begin{proof}
The number $m$ of edges in $G$ is bounded by $(n^2-n)/2$. Hence,
$p = \sum\Set{\degr{}{v}\,|\,v\in V} \leqslant (n^2-n)$, implying
that the total number of vertices $p+n$ in $G_w$ is 
$\leqslant (n^2-n) + n = n^2$.
The number of edges in $K_p$ is $(p^2-p)/2$, and in $K_{p+1}$ it is 
$(p^2+p)/2$, which is $\leqslant \bigl((n^2-n)^2+(n^2-n)\bigr)/2
= (n^4 -2n^3+ 2n^2 -n)/2$.
Hence, the total number of edges in $G_w$ is
$m+ (p^2+p)/2 \leqslant (n^4-2n^3+3n^2-2n)/2$.
\end{proof}

Let $\LL$ be a linear binary tree over $V$ where, as in~\eqref{eq:linear1}
in the opening paragraph of Section~\ref{sect:linear-reassembling}, the longest
chain of nested clusters of size $\geqslant 2$ is: 
\[
   X_1 \ \subsetneq X_2 
                  \ \subsetneq\ \cdots\ \subsetneq\ X_{n-1} = V ,
\]
and let the corresponding singleton clusters be $\Set{Y_0,\ldots,Y_{n-1}}$
as determined by~\eqref{eq:linear2}. The linear tree
$\LL$ is uniquely determined by a sequence of vertices written
in the form:
\[
    [v_1\quad\cdots\quad v_n] 
\]
where $Y_0 = \Set{v_1}, Y_1 = \Set{v_2},\ldots, Y_{n-1} = \Set{v_n}$.
We say $[v_1\ \cdots\ v_n]$ is the
\emph{vertex sequence induced} by $\LL$, and $\LL$ the 
\emph{linear reassembling} (or the \emph{linear binary tree}) \emph{induced}
by  the vertex sequence $[v_1\ \cdots\ v_n]$.

Similarly, if $\varphi : V\to\Set{1,\ldots,n}$ is a linear
arrangement of $V$, then $\varphi$ is uniquely determined by
a sequence of vertices in the same form:
\[
    [v_1\quad\cdots\quad v_n] 
\]
where ${\varphi}^{-1}(1) = v_1, {\varphi}^{-1}(2) = v_2,\ldots, 
{\varphi}^{-1}(n) = v_n$. We say $[v_1\ \cdots\ v_n]$ is the
\emph{vertex sequence induced} by $\varphi$, and $\varphi$
the \emph{linear arrangement induced} by the vertex sequence
$[v_1\ \cdots\ v_n]$.

\medskip
For the auxiliary graph $G_w$, whether we deal with 
a linear reassembling $(G_w,\LL)$ or a linear arrangement 
$(G_w,\varphi)$, it is convenient to consider sequences of the
following form, which interleaves vertices and cutwidths:
\begin{equation*}
\label{eq:sequence-of-vertices-and-cutwidths}
\tag{$\diamondsuit$}
   \sss \ :=\ \bigl[ x_1\quad (r_1,s_1) \quad x_2\quad (r_2,s_2) 
   \quad\cdots\quad\cdots\quad
   x_{n+p-1}\quad (r_{n+p-1},s_{n+p-1}) \quad  x_{n+p} \bigr]
\end{equation*}
where $\Set{x_1,\ldots,x_{n+p}} = V_w = 
       \Set{v_1,\ldots,v_n}\cup\Set{u_1,\ldots,u_p}$,
and for every $1\leqslant i\leqslant n+p-1$:
\begin{alignat*}{5}
  & r_i\ :=\ \degr{G}{\Set{x_1,\ldots,x_i}} 
\quad \text{and} \quad 
   s_i\ :=\ \degr{K_{p+1}}{\Set{x_1,\ldots,x_i}} .
\end{alignat*}
We say the sequence $\sss$
in~\eqref{eq:sequence-of-vertices-and-cutwidths} is 
the \emph{sequence of vertices and cutwidths induced} by
$(G_w,\LL)$ or by $(G_w,\varphi)$, whichever of the two is the case.
The measure $\beta$ on $\sss$ is:
\[
   \beta(\sss)\ :=\ \sum_{1\leqslant i\leqslant n+p-1} (r_i + s_i).
\]

\begin{lemma}
\label{lem:relating-cutwidths}
Consider the sequence of vertices and cutwidths induced
by $(G_w,\LL)$ or by $(G_w,\varphi)$, as just defined.
\textbf{Conclusion:} 
\begin{itemize}[itemsep=0pt,parsep=5pt,topsep=5pt,partopsep=0pt] 
\item  
      For every $1\leqslant i\leqslant n+p-1$, it holds that 
      $r_i+s_i = \degr{G_w}{\Set{x_1,\ldots,x_i}}$.
\item 
      If the sequence in~\eqref{eq:sequence-of-vertices-and-cutwidths} is 
      induced by the linear arrangement $(G_w,\varphi)$, then
      \[
          \beta (G_w,\varphi)\ =\ \beta(\sss)
          \ =\ \sum_{1\leqslant i\leqslant n+p-1} (r_i + s_i) .
      \]
\item 
      If the sequence in~\eqref{eq:sequence-of-vertices-and-cutwidths} is 
      induced by the linear reassembling $(G_w,\LL)$, then
      \[
          \beta (G_w,\LL)\ =\ \Delta\ +\ \beta(\sss) 
          \ =\ \Delta\ +\ \sum_{1\leqslant i\leqslant n+p-1} (r_i + s_i) ,
      \]
      where $\Delta = \sum \Set{\degr{G_w}{v}\,|\,v\in V_w\text{ and } v\neq x_1}$.
\end{itemize}
\end{lemma}

\begin{proof}
Straightforward consequence of the definitions. All details omitted.
\end{proof}

We say that the sequence $\sss$ 
is \emph{scattered} if the vertices of $K_{p+1}$ do not occur consecutively,
\ie, the vertices of $K_{p+1}$ are interspersed with vertices of $V-\Set{w}$.

\begin{lemma}
\label{lem:auxiliary-vertices-together-1}
Let $G_w = (V_w,E_w)$ be the auxiliary graph for vertex $w\in V$, as
constructed in Definition~\ref{def:auxiliary-graphs}.  
Let $\sss$ be the sequence of vertices and cutwidths, 
as in~\eqref{eq:sequence-of-vertices-and-cutwidths},
induced by a $\beta$-optimal linear reassembling $(G_w,\LL)$ or 
by a $\beta$-optimal linear arrangement $(G_w,\varphi)$.
\textbf{Conclusion:} $\sss$ is not scattered.
\end{lemma}

\medskip
In words, in a $\beta$-optimal linear reassembling $(G_w,\LL)$ [or
in a $\beta$-optimal linear arrangement $(G_w,\varphi)$, resp.] all
the vertices of $K_{p+1}$ are reassembled consecutively [or arranged
consecutively, resp.] without intervening vertices from $V-\Set{w}$.
 
\begin{proof}
In Appendix~\ref{sect:remaining-proofs-for-linear}.
\end{proof}

Consider again the sequence $\sss$ of vertices and cutwidths
in~\eqref{eq:sequence-of-vertices-and-cutwidths}. Suppose $\sss$ is
not scattered.  This means that the $p+1$ vertices of $K_{p+1}$ occur
consecutively in $\sss$. We say $\sss$ is \emph{balanced}
iff one of two conditions holds:
%
%
\begin{alignat*}{8}
   &(1) \qquad && \Set{x_1,\ldots,x_{n-1}}\ &&=\ && V - \Set{w},
     \qquad && \Set{x_n}\ && =\ && \Set{w},\qquad 
     && \Set{x_{n+1},\ldots,x_{n+p}} = U, 
\\
   & (2)\qquad && \Set{x_1,\ldots,x_{p}}\ &&=\ && U,
     \qquad && \Set{x_{p+1}}\ &&=\ && \Set{w},\qquad 
    && \Set{x_{p+2},\ldots,x_{n+p}} = V -\Set{w}.
\end{alignat*}
In words, $\sss$ is balanced if all the vertices of $V-\Set{w}$ are
on the same side (on the left in (1), or on the right in (2)) of 
the distinguished vertex $w$
\text{and} all the vertices of $U$ are on the other side 
(on the right in (1), or on the left in (2)) of $w$. 
Put differently still, $\sss$ is balanced if
all the vertices of $V-\Set{w}$ are together, all the vertices of 
$U$ are together, and $w$ is between the two sets of vertices.
The following is a refinement of the preceding lemma and its proof
is based on a similar argument.

\begin{lemma}
\label{lem:auxiliary-vertices-together-2}
Let $G_w = (V_w,E_w)$ be the auxiliary graph for vertex $w\in V$, as
constructed in Definition~\ref{def:auxiliary-graphs}.  
Let $\sss$ be the sequence of vertices and cutwidths, 
as in~\eqref{eq:sequence-of-vertices-and-cutwidths},
induced by a $\beta$-optimal linear reassembling $(G_w,\LL)$ or 
by a $\beta$-optimal linear arrangement $(G_w,\varphi)$.
\textbf{Conclusion:} ${\sss}$ is balanced.
\end{lemma}

\begin{proof}
In Appendix~\ref{sect:remaining-proofs-for-linear}.
\end{proof}

By the preceding lemma, if the sequence $\sss$ 
in~\eqref{eq:sequence-of-vertices-and-cutwidths} is induced
by a $\beta$-optimal linear reassembling $(G_w,\LL)$, or
by a $\beta$-optimal linear arrangement $(G_w,\varphi)$, then
$\sss$ is balanced, either on the left or on the right. For 
the rest of the analysis below, we assume that $\sss$ is balanced
on the right, \ie, all the vertices in $U$ occur first, then $w$,
and then all the vertices of $V-\Set{w}$.

\begin{definition}{Restrictions of linear reassemblings and linear
arrangements}
\label{def:restricted-linear}
Let $\LL$ be a linear binary tree over the set $V$.
If $V' \subseteq V$, the \emph{restriction} of $\LL$ to $V'$,
denoted $(\LL\,|\,V')$, consists of the following clusters:
\[
    (\LL\,|\,V')\ :=\ \SET{\, X\cap V'\;\bigl|\; X\in \LL\,} .
\]
It is a straightforward exercise to show that $(\LL\,|\,V')$
is a linear binary tree over $V'$.

Let $\varphi : V\to\Set{1,\ldots,n}$ be a linear arrangement
of $V$. The \emph{restriction} of $\varphi$ to $V'$, 
denoted $\bigl(\varphi\,|\,V'\bigr)$, is defined as follows.
For every $1\leqslant i \leqslant n'= \size{V'}$, let:
\begin{alignat*}{5}
   & \bigl(\varphi\,|\,V'\bigr)\, (v)\ &&:=\ && i
     \quad && \text{where $v = {\varphi}^{-1}(j)$ and $j\in\Set{1,\ldots,n}$ is}
\\ & && && && \text{the largest integer such that 
      $\ssize{\Set{{\varphi}^{-1}(1),\ldots,{\varphi}^{-1}(j-1)}\cap V'} = i-1$.}
\end{alignat*}
Again here, it is straightforward to show that $\bigl(\varphi\,|\,V'\bigr)$
is a linear arrangement of $V'$ such that:
\[
  \bigl(\varphi\,|\,V'\bigr)^{-1}(1),\ \ldots\ ,
  \ \bigl(\varphi\,|\,V'\bigr)^{-1}(n')
  \quad\text{is a subsequence of}\quad
  {\varphi}^{-1}(1),\ \ldots\ ,\ {\varphi}^{-1}(n) .
\]
Moreover, if $(G,\LL)$ is a linear reassembling 
[resp. $(G,\varphi)$ is a linear arrangement] of 
the simple undirected graph $G = (V,E)$ and $G' = (V',E')$
is the subgraph of $G$ induced by $V'\subseteq V$, then 
$\bigl(G',(\LL\,|\,V')\bigr)$ is a linear reassembling
[resp. $\bigl(G',(\varphi\,|\,V')\bigr)$ is a linear arrangement]
of $G'$.
\end{definition}

\begin{lemma}
\label{lem:beta-optimal-reassembling}
Let $G_w = (V_w,E_w)$ be the auxiliary graph for vertex $w\in V$, as
constructed in Definition~\ref{def:auxiliary-graphs}. 
\begin{enumerate}[itemsep=0pt,parsep=2pt,topsep=5pt,partopsep=0pt] 
\item If $(G_w,\LL)$
      is a $\beta$-optimal linear reassembling of $G_w$ with no anchor
      restriction, then 
      $\bigl(G,(\LL\,|\,V)\bigr)$ is a $\beta$-optimal 
      linear reassembling relative to anchor $w$.
\item If $(G_w,\varphi)$
      is a $\beta$-optimal linear arrangement of $G_w$ with no anchor
      restriction, then 
      $\bigl(G,(\varphi\,|\,V)\bigr)$ is a $\beta$-optimal 
      linear arrangement relative to anchor $w$.
\end{enumerate} 
\end{lemma}

\begin{proof}
We prove part 1 only, the proof of part 2 is similar.
By Lemma~\ref{lem:auxiliary-vertices-together-2}, the
sequence $\sss$ induced by a $\beta$-optimal linear reassembling
$(G_w,\LL)$ is balanced. By our assumption preceding 
Definition~\ref{def:restricted}, we take $\sss$ to be balanced on the
right, \ie, all the vertices in $U$ occur first, then $w$,
and then all the vertices of $V-\Set{w}$. There are no edges
connecting vertices in $U$ on the left to vertices in $V-\Set{w}$ on the 
right, with $w$ a cut vertex in the middle. The $\beta$-optimality of 
$(G_w,\LL)$ implies the $\beta$-optimality of the linear
reassembling $\bigl(G,(\LL\,|\,V)\bigr)$ of the subgraph 
$G = (V,E)$ of  $G_w = (V_w,E_w)$. We omit all formal details.
\end{proof}

\begin{theorem}
\label{thm:reducing-linear-reassemblings-to-linear-arrangements-2}
For the class of all simple undirected graphs $G$, the two following
problems are polynomial-time reducible to each other:
\begin{itemize}[itemsep=0pt,parsep=2pt,topsep=5pt,partopsep=0pt] 
\item the $\beta$-optimality of linear arrangements $(G,\varphi)$,
\item the $\beta$-optimality of linear reassemblings $(G,\LL)$.
\end{itemize}
\end{theorem}

\noindent
More explicitly, a polynomial-time algorithm $\A$, which 
returns a $\beta$-optimal linear reassembling $(G,\LL)$ 
[resp. a $\beta$-optimal linear arrangement $(G,\varphi)$]
of an arbitrary graph $G$, can be used to return a $\beta$-optimal 
linear arrangement $(G,\varphi)$ [resp. a $\beta$-optimal 
linear reassembling $(G,\LL)$] in polynomial time.

\begin{proof}
We compute a $\beta$-optimal linear reassembling $(G_{v_i},{\LL}_i)$ 
[resp. a $\beta$-optimal linear arrangement $(G_{v_i},{\varphi}_i)$] of the
auxiliary graph $G_{v_i}$, one for each vertex 
$v_i\in V = \Set{v_1,\ldots,v_n}$. We next consider
the linear reassembling $(G,(\LL_i\,|\,V))$ 
[resp. the linear arrangement $(G,(\varphi_i\,|\,V))$] which,
by Lemma~\ref{lem:beta-optimal-reassembling}, is a $\beta$-optimal linear
reassembling  relative to anchor $v_i$
[resp. a $\beta$-optimal linear arrangement relative to anchor $v_i$], 
for every $1\leqslant i\leqslant n$. Let $(G,\varphi_i)$ be the linear
arrangement induced by the linear reassembling $(G,(\LL_i\,|\,V))$
[resp. let $(G,\LL_i)$ be the linear reassembling induced by the 
linear arrangement $(G,(\varphi_i\,|\,V))$].
By Lemma~\ref{lem:equivalence-when-starting-vertex-is-the-same},
$(G,\varphi_i)$ is a $\beta$-optimal linear arrangement relative to anchor $v_i$
[resp. $(G,\LL_i)$ is a $\beta$-optimal linear reassembling  
relative to anchor $v_i$], for every $1\leqslant i\leqslant n$. 
Among these $n$ linear arrangements
[resp. $n$ linear reassemblings], 
we choose one such that $\beta(G,\varphi_i)$ is minimized
[resp. $\beta(G,\LL_i)$ is minimized].
\end{proof}

\begin{corollary}
\label{cor:NP-hardness-of-beta-optimality-of-linear-reassemblings}
For the class of all simple undirected graphs $G$,
the computation of $\beta$-optimal linear reassemblings $(G,\LL)$
is an NP-hard problem.
\end{corollary}

\begin{proof}
This follows from the NP-hardness of 
the \emph{minimum-cost linear arrangement} problem
(also called the \emph{optimal linear arrangement} 
problem in the literature)~\cite{gareyJohnsonStockmeyer1976}.
This problem is the same as what we call, in this report, the problem
of computing a $\beta$-optimal linear arrangement.
\end{proof}

\begin{remark}
\label{rem:open-problem2}
To the best of our knowledge, the complexity status of
the \emph{minimum-cost linear arrangement} problem (or \emph{optimal
linear arrangement} problem) for $k$-regular graphs for a fixed
$k\geqslant 3$ is an open problem. If it were known to be NP-hard, we
would be able to simplify our proof of
Theorem~\ref{thm:reducing-linear-reassemblings-to-linear-arrangements-2}
and its corollary considerably.
\end{remark}

\Hide{
\begin{corollary}
\label{cor:NP-hardness-of-relative-beta-optimality}
For the class of all simple undirected graphs $G = (V,E)$,
each with a distinguished vertex $w\in V$, the following
problems are NP-hard:
\begin{enumerate}[itemsep=0pt,parsep=2pt,topsep=5pt,partopsep=0pt] 
\item The computation of linear reassemblings $(G,\LL)$ which
        are $\beta$-optimal relative to anchor $w$.
\item The computation of linear arrangements $(G,\varphi)$ which
        are $\beta$-optimal relative to anchor $w$.
\end{enumerate}
\end{corollary}

\begin{proof}
It suffices to prove part 1, after which part 2 is a consequence of
Theorem~\ref{thm:equivalence-of-anchored-beta-quivalence}.  For part
1, we show that the problem of computing a $\beta$-optimal linear
reassembling $(G,\LL)$ with no anchor restriction is reducible to the
problem of computing a linear reassembling $(G,\LL)$ which is
$\beta$-optimal relative to an anchor $w$. The argument is similar to
that in the proof of
Theorem~\ref{thm:reducing-linear-reassemblings-to-linear-arrangements-2},
where $w$ is set to one of the $n$ vertices of $G$, consecutively $n$
different times. But the problem of computing a $\beta$-optimal linear
reassembling $(G,\LL)$ with no anchor restriction is NP-hard, by
Corollary~\ref{cor:NP-hardness-of-beta-optimality-of-linear-reassemblings},
which implies the desired conclusion.
\end{proof}
}


\section{Related and Future Work}
\label{sect:future}

We mentioned several open problems from the literature, still
unresolved to the best of our knowledge, in
Remarks~\ref{rem:open-problem1}, ~\ref{rem:anchored},
and~\ref{rem:open-problem2}. If these open problems were solved,
partially or optimally, they would permit various simplifications in
our proofs. In particular, even though one of our reductions can be
carried out in polynomial time by invoking an earlier result on cutwidths
(Lemmas~\ref{lem:thilikos1} and~\ref{lem:thilikos2}),
its $\bigOO{n^{12}}$ complexity is
prohibitive (see the proof of
Theorem~\ref{thm:reducing-linear-arrangements-to-linear-reassemblings});
this is the reduction that reduces the $\alpha$-optimality of linear
arrangements to the $\alpha$-optimality of linear reassemblings.

Beyond open problems whose resolutions would simplify and/or
strengthen some of this report's results and their proofs, our wider research
agenda is to tackle forms of graph reassembling other
than \emph{linear} -- in particular,
\emph{balanced reassembling} and \emph{binary reassembling} in
general, both \emph{strict} and \emph{non-strict}, all alluded to in
Sections~\ref{sect:intro}, \ref{sect:problem},
and~\ref{sect:examples}.  For each form of reassembling, both
$\alpha$-optimization and $\beta$-optimization will have to be
addressed; as suggested by the examination in this report, these two
optimizations seem to call for different proof methods, despite their
closely related definitions.

We also need to study classes of graphs for which
$\alpha$-optimization and/or $\beta$-optimization of their 
reassembling, in any of the forms mentioned above, can be carried
out in low-degree polynomial times.
Finally, there is the question of whether, by
allowing \emph{approximate solutions}, we can turn the NP-hardness of
any of the preceding optimizations into polynomially-solvable
optimizations. The literature on approximation algorithms dealing with
graph layout problems is likely to be an important resource to draw
from (among many other papers, 
the older~\cite{arora1996new,leighton1999multicommodity,rao1998new}
the more recent~\cite{charikar20062}, and the survey~\cite{petit:2011}).

\newpage

\appendix
 \section{Appendix: Sequential Graph Reassembling}
    \label{sect:sequential}

Let $\PPP$ be the set of all the partitions of the set 
$V = \Set{v_1,\ldots,v_n}$ of
vertices in the graph $G = (V,E)$.  There are two special partitions in 
$\PPP$:
\[  P_0 := \SET{\,\Set{v}\;\bigl|\;v\in V\,}
    \quad\text{and}\quad P_{\infty} := \Set{\,V\,} .
\]
Given two partitions $X, Y \in\PPP$, we write $X \sqsubseteq Y$
if $X$ is \emph{finer} than $Y$ or, equivalently, $Y$
is \emph{coarser} than $X$, \ie, for every block $A\in X$ there is a
block $B\in Y$ such that $A\subseteq B$. We write $X\sqsubset\Y$
iff $X \sqsubseteq Y$ and $X\neq Y$. The relation
``$\sqsubseteq$'' is a (non-strict) partial order on $\PPP$,
with a least element (the \emph{finest} partition $P_0$) and a
largest element (the \emph{coarsest} partition $P_{\infty}$).
We need the following simple fact.

\begin{lemma}
\label{lem:maximal-partition-chain}
In the poset $(\PPP,\sqsubseteq)$ of all partitions of
$n$ elements, a maximal chain (linearly ordered with 
$\sqsubset$) is a sequence of $n$ partitions, always
starting with $P_0$ and ending with $P_{\infty}$.
\end{lemma}

\begin{proof}
If $X_1\sqsubset X_2\sqsubset\cdots\sqsubset X_k$ is a maximal
chain of partitions, necessarily with $X_1 = P_0$ and $X_k = P_{\infty}$
because the chain is maximal,
then $X_1$ has $n$ blocks, $X_2$ has $n-1$ blocks, in general
$X_p$ has $n-p+1$ blocks, and $X_k$ has one block. The length
$k$ of the chain is therefore exactly $n$.
\end{proof}

\begin{definition}{Sequential graph reassembling, \ie, according to an
ordering of the edges}
\label{defn:sequential-graph-reassembling}
Let $\EdgePerm$ be an ordering of the edges in $E$. We use $\EdgePerm$
to select $n$ partitions in $\PPP$ forming a maximal chain 
(linearly ordered with $\sqsubset$), 
which starts with the finest partition $P_0$ and ends with the
coarsest partition $P_{\infty} = \Set{V}$, say:
\[
   X_1\ \sqsubset\ X_2\ \sqsubset
   \ \cdots\ \sqsubset X_n \qquad
   \text{where $X_1 = P_0$ and $X_n = P_{\infty}$},
\]
as we explain next. To define $X_{p+1}$ from $X_p$,
we associate each $X_p$ with a subsequence $\EdgePerm_p$ of 
the initial sequence $\EdgePerm_1 = \EdgePerm$, 
for every $p\geqslant 1$. 
The subsequence ${\EdgePerm}_p$ keeps track of all the edges that have not yet 
been reconnected.
We obtain the next pair $\bigl(X_{p+1},{\EdgePerm}_{p+1}\bigr)$
from the preceding pair $\bigl(X_{p},{\EdgePerm}_p\bigr)$ as follows: 
\begin{itemize}[itemsep=0pt,parsep=2pt,topsep=5pt,partopsep=0pt] 
\item[(1)] Take the first edge $e$ in the sequence ${\EdgePerm}_p$,
           \ie, let ${\EdgePerm}_p = e\,{\EdgePerm}_p'$ for some 
           ${\EdgePerm}_p'$, with $e = \set{v\,w}$ for some 
           $v,w\in V$, and let $A$ and $B$ be the unique blocks in $X_{p}$ 
           containing $v$ and $w$, respectively.           
\item[(2)] Merge the two blocks $A$ and $B$ to obtain $X_{p+1}$, \ie, let:
           \[
            X_{p+1} := \bigl(X_{p} - \Set{A,B}\bigr)\cup \Set{ A\cup B }.
           \]
\item[(3)] Delete every edge $e'$ whose two endpoints are in 
           the new block $A\cup B$
           to obtain ${\EdgePerm}_{p+1}$, \ie, let:
           \[ {\EdgePerm}_{p+1} :=  {\EdgePerm}_{p}\;\bigl/
            \;\SET{\,e'\in E\;\bigl|\;e'=\set{v'\,w'}\text{ and } \Set{v',w'}
            \subseteq A\cup B\,} .
           \]
\end{itemize}
In words, we go from $\bigl(X_{p},{\EdgePerm}_p\bigr)$ to 
$\bigl(X_{p+1},{\EdgePerm}_{p+1}\bigr)$ by merging the two blocks $A$ and $B$
in $X_{p}$ that are connected by the first edge $e$ in ${\EdgePerm}_p$, and then
removing from further consideration all edges whose endpoints are
in $A\cup B$.  

We refer to the sequential reassembling of $G$ according to
the ordering ${\EdgePerm}$ by writing $(G,{\EdgePerm})$,
the result of which is the chain of partitions 
$\X = X_1\sqsubset\cdots\sqsubset X_n$, more succintly
written also as $\X = X_1\,\cdots\,X_n$.
\end{definition}

\begin{remark}
In Definition~\ref{defn:sequential-graph-reassembling}, 
when we merge the two blocks $A$ and $B$
because the edge $e$ has its two endpoints in $A$ and $B$, 
not only do we reconnect the two halves of $e$, but 
we additionally reconnect every other edge $e'$ whose two endpoints
are also in $A$ and $B$. Thus, in general, we may reconnect several edges
simultaneously -- all the edges between $A$ and $B$ in the original
graph -- rather than one at a time by strictly following the
order specified by ${\EdgePerm}$. 
The same happens with binary graph-reassembling
(Definition~\ref{defn:binary-reassembling}).
\end{remark}

Definition~\ref{defn:sequential-graph-reassembling} describes the
process of going from an ordering ${\EdgePerm}$ of edges to a maximal
chain $\X$ of partitions. If $\X = X_1\,\cdots\,X_n$ is
a maximal chain of partitions, we say that $\X$ is \emph{strict} if, for
every consecutive pair $(X_p,X_{p+1})$ with $A,B\in X_p$ and $A\cup
B\in X_{p+1}$, where $1\leqslant p < n$, it is the case that
$\bridges{G}{A,B}\neq\varnothing$. The result of a sequential
reassembling $(G,\EdgePerm)$ is always a strict maximal chain $\X$
of partitions. We can also carry out the process in
reverse, as asserted by the next lemma. 

\begin{lemma}[From a maximal chain $\X$ of partitions  to
an ordering ${\EdgePerm}$ that induces it]
\label{prop:from-a-maximal-partition-chain}
\label{lem:from-a-maximal-partition-chain}
Let $G = (V,E)$ be a graph, and $\PPP$ the set of all partitions of $V$,
as in Definition~\ref{defn:sequential-graph-reassembling}.
For every maximal chain of partitions 
$\X = X_1\sqsubset\cdots\sqsubset X_n$, with $X_1,\ldots,X_n\in\PPP$,
if $\X$ is strict, then 
there is an ordering (not necessarily unique) ${\EdgePerm}$ of $E$ such that 
$(G,{\EdgePerm}) = \X$.
\end{lemma}

\begin{proof}
This is a consequence of 
Definition~\ref{defn:sequential-graph-reassembling}.
Details omitted.
\end{proof}

We want to relate the two notions:
sequential reassembling $(G,{\EdgePerm})$ in
Definition~\ref{defn:sequential-graph-reassembling} and 
binary reassembling $(G,{\B})$ in Definition~\ref{defn:binary-reassembling}. 
The discussion to follow uses the following facts.

\begin{lemma}
\label{lem:maximal-cross-section-chain}
Let $\B$ be a binary tree over $V = \Set{v_1,\ldots,v_n}$, given in the
formulation of Definition~\ref{defn:binaryTrees}. 
\begin{enumerate}[itemsep=0pt,parsep=2pt,topsep=5pt,partopsep=0pt] 
\item If $S = \Set{X_1,\ldots,X_k} \subseteq\B$
    is a maximal collection of $k\geqslant 2$ pairwise disjoint sets
    in $\B$, then $S$ is a partition of $V$. We call such a maximal 
    collection $S$ a \emph{cross-section} of the binary tree $\B$.
\item If $\SSS$ is the set of all cross-sections of $\B$, then
    $(\SSS,\sqsubseteq)$ is a proper sub-poset of the
    poset $(\PPP,\sqsubseteq)$ in Lemma~\ref{lem:maximal-partition-chain},
    with the same bottom element $P_0$ and top element $P_{\infty}$.
\item In the poset $(\SSS,\sqsubseteq)$, a maximal chain
    has exactly $n$ entries, always starting with $P_0$ and ending 
    with $P_{\infty}$.
\end{enumerate}
\end{lemma} 

\begin{proof}
All three parts can be proved by induction
on $n \geqslant 1$, using the same reasoning as in the
proofs of Propositions~\ref{prop:propertiesOne} and~\ref{prop:propertiesTwo}. 
All details omitted.
\end{proof}

In the preceding lemma, it is worth noting that the size of $\PPP$ is
fixed as a function of $n$, the so-called \emph{Bell number}
$B(n)$, which counts the partitions of an $n$-element set and grows
exponentially in $n$.%
   \footnote{There is no known simple expression for the exponential
   growth of $B(n)$ as a function of $n$, though there are various
   ways of estimating tight lower bounds and tight upper bounds on its
   asymptotic growth \cite{Odlyzko95asymptoticenumeration}.
   }
By contrast, the size of $\SSS$ is much smaller, depends on both $n$
and the shape of the binary tree $\B$, and can be as small as $n$ (the
case when $\B$ is a linear, \ie, a degenerate binary tree).

\medskip
Consider a sequential reassembling 
$(G,\EdgePerm)$ of the graph $G = (V,E)$, the result of which is a 
maximal chain of $n$ partitions $\X = X_1\sqsubset\ \cdots\ \sqsubset X_n$,
as in Definition~\ref{defn:sequential-graph-reassembling}, where
$X_1 = P_0$ and $X_n = P_{\infty}$. Since every successive partition $X_{p+1}$
in $\X$ is obtained from the previous $X_p$ by merging two blocks in
$X_p$, there is a natural way of organizing $\X$ in the form of a
binary tree $\B$, with $n$ (not all)
of the cross-sections of $\B$ being exactly 
$\Set{X_1,\ldots,X_n}$. Let $\binary{G,\EdgePerm}$
denote the binary reassembling thus obtained. 

\medskip
Consider next a binary reassembling 
$(G,\B)$ of the graph $G = (V,E)$. The set $\SSS$ of all cross-sections
in $\B$ is uniquely defined. We want to extract from $\SSS$ a maximal 
chain $\X$ of cross-sections/partitions, ordered by $\sqsubset$, which,
by Lemma~\ref{lem:from-a-maximal-partition-chain},
will in turn induce an ordering of $\EdgePerm$ of the edges. The
problem here is that there are generally many such maximal chains $\X$.
We need therefore a method to canonically extract a unique maximal chain
$\X$ from $\SSS$ and a unique edge-ordering $\EdgePerm$ from $\X$.
We propose such a method in the next paragraph.

We assume that the binary reassembling $(G,\B)$ is strict and that
there is a fixed ordering of the vertices, say, $v_1 \prec
v_2 \prec \cdots \prec v_n$.  The vertex ordering ``$\prec$'' is
extended to edges, and to sets of edges, as follows:
\begin{itemize}[itemsep=0pt,parsep=2pt,topsep=5pt,partopsep=0pt] 
\item If $e = \set{v\,w}$ is the edge joining vertices $v$ and $w$, 
      we assume $v\prec w$.
\item If $e = \set{v\,w}$ and $e' = \set{v'\,w'}$, then
      $e \prec e'$ iff \emph{either} $v\prec v'$ \emph{or}
      $v = v'$ and $w\prec w'$.
\item If $A \subseteq E$, then $\canonicalOrd{A}$ is the 
      \emph{canonical ordering}
      of $A$ w.r.t. ``$\prec$'', \ie, $\canonicalOrd{A} = e_1\,e_2\,\cdots\,e_k$\\
      where $A = \Set{e_1,e_2,\ldots,e_k}$ and 
      $e_1\prec e_2\prec \cdots\prec e_k$.
\item If $A$ and $B$ are non-empty disjoint set of edges,
      with $\canonicalOrd{A} = e_1\,e_2\,\cdots$ and 
      $\canonicalOrd{B} = f_1\,f_2\,\cdots$, \\
      then $\canonicalOrd{A}\prec \canonicalOrd{B}$ iff $e_1\prec f_1$.
\end{itemize}
If $W\in\B$, then ${\B}_W$ is the subtree of $\B$ rooted at
$W$ (see Proposition~\ref{prop:propertiesTwo}). We write 
$(G,{\B}_W)$ for a \emph{partial} binary reassembling of the graph $G=(V,E)$,
the result being the subgraph of $G$ induced by $W$ together with all
the edges in $\bridges{G}{W}$ as dangling edges, \ie, edges with only
one endpoint in $W$. We define a canonical ordering of all the edges
already in place in the partial reassembling $(G,{\B}_W)$, denoted 
$\canonicalOrd{G,{\B}_W}$, as follows:
\[
  \canonicalOrd{G,{\B}_W}\;:=
  \ \begin{cases}
    \varepsilon\ (\text{the empty string}) 
      & \text{if $W$ is a singleton set},
    \\[1.2ex]
    \canonicalOrd{G,{\B}_{T}}\,\canonicalOrd{G,{\B}_{U}}
    \,\canonicalOrd{\bridges{}{T,U}}
      & \text{if $W = T\uplus U$ and},
    \\
      & \text{$\canonicalOrd{G,{\B}_{T}}\prec\canonicalOrd{G,{\B}_{U}}$},
    \\[1.2ex]
    \canonicalOrd{G,{\B}_{U}}\,\canonicalOrd{G,{\B}_{T}}
    \,\canonicalOrd{\bridges{}{T,U}}
      & \text{if $W = T\uplus U$ and},
    \\
      & \text{$\canonicalOrd{G,{\B}_{U}}\prec\canonicalOrd{G,{\B}_{T}}$}.
    \end{cases}
\]
Because the binary reassembling $(G,\B)$ is strict, 
$\bridges{}{T,U}\neq\varnothing$ in the second and third cases above,
which implies $\canonicalOrd{\bridges{}{T,U}}\neq\varepsilon$. 
If $W = V$, then $(G,{\B}) = (G,{\B}_W)$ and 
$\canonicalOrd{G,{\B}} = \canonicalOrd{G,{\B}_W}$.

\begin{proposition}[Relating sequential reassembling and binary reassembling]
\label{prop:relating-two-reassembling}
Let $G= (V,E)$ be a simple undirected graph. We have the following facts:
\begin{enumerate}[itemsep=0pt,parsep=2pt,topsep=2pt,partopsep=0pt] 
\item For every sequential reassembling $(G,{\EdgePerm})$,
      there is a binary tree $\B$ over $V$ such that:\\
      \(
        \binary{G,{\EdgePerm}} = (G,\B) .
      \) 
\item  For every strict binary reassembling $(G,\B)$,
    there is an ordering ${\EdgePerm}$ of $E$ such that:\\
      \(
      \canonicalOrd{G,\B} = (G,\EdgePerm) .
      \)
\item For every strict binary reassembling $(G,\B)$, it holds that:
      \(
      \binary{\canonicalOrd{G,\B}} = (G,\B) .
      \)
\end{enumerate}
\end{proposition}

\begin{proof} Parts 1 and 2 follow from the definitions and discussion
that precede the proposition. All details omitted. Part 3 can be proved 
by structural induction on the subtrees ${\B}_W$ of $\B$, where
the induction hypothesis is $\binary{\canonicalOrd{G,{\B}_W}} = (G,{\B}_W)$.
All details omitted again.
\end{proof}

It is possible to refine the notion of ``canonical ordering'' on the
set of edges $E$, so that the equality 
$\canonicalOrd{\binary{G,{\EdgePerm}}} = (G,{\EdgePerm})$ holds which,
together with the equality in part 3 of the preceding proposition, will
mean that the functions $\binary{}$ and $\canonicalOrd{}$ are inverses of
each other. We omit this refinement as it will go further afield from our 
main concerns.

\newpage

\section{Appendix: Remaining Proofs for Sections~\ref{sect:alpha-optimization}
  and~\ref{sect:beta-optimization} }
   \label{sect:remaining-proofs-for-linear}

We supply the details of several long straightforward  
and/or highly technical proofs which we omitted in 
Sections~\ref{sect:alpha-optimization} and~\ref{sect:beta-optimization} 
in order to facilitate the grasp of the different concepts and 
their mutual dependence.

\paragraph{Proof of 
Theorem~\ref{thm:reducing-linear-reassemblings-to-linear-arrangements}.}
\label{proof:reducing-linear-reassemblings-to-linear-arrangements} 
Let $G = (V,E)$ be an arbitrary simple undirected graph, with
$\size{V} = n$.  It suffices to show that if $(G,\varphi)$ is an
$\alpha$-optimal linear arrangement, then the linear reassembling
$(G,\LL)$ induced by $(G,\varphi)$ is also $\alpha$-optimal,
using Definition~\ref{def:linear-reassembling-induced}.

In the notation of Definition~\ref{def:linear-reassembling-induced},
the clusters of $\LL$ of size $\geqslant 2$ are $\Set{X_1,\ldots,X_{n-1}}$.
For the singleton clusters of $\LL$, we
pose $Y_{i-1} := \Set{ {\varphi}^{-1}(i) }$,
where $1\leqslant i\leqslant n$. From Definition~\ref{def:linear-arrangement}:
\begin{alignat*}{7}
   & \alpha(G,\varphi)\ &&=\ &&
   \max\,\SET{\,\degr{}{Y_0},
     \;\max\,\Set{\,\degr{}{X_j}\;|\;1\leqslant j\leqslant n-1}} ,
\\
   & \alpha(G,\LL)\ &&=\ &&
   \max\,\SET{\,\max\,\Set{\,\degr{}{Y_i}\;|
         \;0\leqslant j\leqslant n-1},
     \;\max\,\Set{\,\degr{}{X_j}\;|\;1\leqslant j\leqslant n-1}} .
\end{alignat*}
By way of getting a contradiction, assume that 
$(G,\varphi)$ is $\alpha$-optimal but that the induced $(G,\LL)$ is not
$\alpha$-optimal. Hence, there is another linear reassembling
$(G,\LL')$ which is $\alpha$-optimal such that
$\alpha(G,\LL') < \alpha(G,\LL)$. Using the same notation for
both $(G,\LL)$ and $(G,\LL')$, where every name related to the latter is
decorated with a prime, the inequality
$\alpha(G,\LL') < \alpha(G,\LL)$ implies the inequality:
\begin{alignat*}{7}
   & \max\,\SET{\,\max\,\Set{\,\degr{}{Y'_i}\;|
         \;0\leqslant j\leqslant n-1},
     \;\max\,\Set{\,\degr{}{X'_j}\;|\;1\leqslant j\leqslant n-1}}
   \\
   & < 
   \ \max\,\SET{\,\max\,\Set{\,\degr{}{Y_i}\;|
         \;0\leqslant j\leqslant n-1},
     \;\max\,\Set{\,\degr{}{X_j}\;|\;1\leqslant j\leqslant n-1}} .
\end{alignat*}
But $\max\Set{\degr{}{Y'_i}\,|\,0\leqslant j\leqslant n-1} =
\max\Set{\degr{}{Y_i}\,|\,0\leqslant j\leqslant n-1}$, which
implies two inequalities:
\begin{alignat*}{7}
\label{eq:unequal1} 
\tag{$1$}
   &  \max\,\Set{\,\degr{}{Y_i}\;|\;1\leqslant j\leqslant n-1}
     \ &&<\ &&
     \max\,\Set{\,\degr{}{X_j}\;|\;1\leqslant j\leqslant n-1} ,
  \\
\label{eq:unequal2} 
\tag{$2$}
  &  \max\,\Set{\,\degr{}{X'_j}\;|\;1\leqslant j\leqslant n-1}
     \ &&<\ &&
     \max\,\Set{\,\degr{}{X_j}\;|\;1\leqslant j\leqslant n-1} .
\end{alignat*}
Hence, by inequality~\eqref{eq:unequal1}, we have:
\[
   \alpha(G,\varphi)\ =\ \alpha(G,\LL)\ =
   \ \max\,\Set{\,\degr{}{X_j}\;|\;1\leqslant j\leqslant n-1} .
\]
Consider now the linear arrangement $(G,\varphi')$ induced by the linear
reassembling $(G,\LL')$, using Definition~\ref{def:linear-arrangement-induced}.
We have:
\[
  \alpha(G,\varphi')\ =\ 
  \max\,\SET{\,\degr{}{Y'_0},
     \;\max\,\Set{\,\degr{}{X'_j}\;|\;1\leqslant j\leqslant n-1}} .
\]
If $\degr{}{Y'_0} \geqslant 
   \max\Set{\degr{}{X'_j}\,|\,1\leqslant j\leqslant n-1}$, then
inequality~\eqref{eq:unequal1} implies $\alpha(G,\varphi') < \alpha(G,\varphi)$,
else inequality~\eqref{eq:unequal2} implies 
again $\alpha(G,\varphi') < \alpha(G,\varphi)$. In both cases,
the $\alpha$-optimality of $(G,\varphi)$ is contradicted. 
\hfill \QED

\bigskip
For the proofs of Lemma~\ref{lem:auxiliary-vertices-together-1} and
Lemma~\ref{lem:auxiliary-vertices-together-2}, we take a closer look
at how the vertices of $K_{p+1}$ are positioned in the sequence $\sss$
in~\eqref{eq:sequence-of-vertices-and-cutwidths} in 
Section~\ref{sect:beta-optimization}. From the fact that
$p$ is the sum of all the vertex degrees in $G$, it follows that $p$
is even and $p+1$ odd. From the sequence $\sss$, we can extract the
subsequence $(\sss\;|\;K_{p+1})$ consisting of all the vertices of
$K_{p+1}$ and corresponding cutwidths:
\begin{equation*}
\label{eq:sequence-of-vertices-and-cutwidths-1}
\tag{$\diamondsuit\diamondsuit$}
   \qquad
   (\sss\;|\;K_{p+1})\ =\ \bigl[ x_{i_1}\quad s_{i_1} \quad x_{i_2}\quad s_{i_2}  
   \quad\cdots\quad\cdots\quad
   x_{i_{p}}\quad s_{i_{p}} \quad  x_{i_{p+1}} \bigr]
\end{equation*}
where $\Set{i_1,\ldots,i_{p+1}}\subseteq\Set{1,\ldots,n+p}$
and $\Set{x_{i_1},\ldots,x_{i_{p+1}}} = \Set{u_1,\ldots,u_p}\cup\Set{w}$.  
In the preceding sequence, every vertex has the same degree $p$ in the
subgraph $K_{p+1}$.  In the full graph $G_w$, every vertex from $K_{p+1}$
has again the same degree $p$, except for the distinguished vertex $w$ 
which has degree $p+d$ where $d = \degr{G}{w}$. In particular, we have:
\[
    s_{i_1} = p, \quad s_{i_2} = 2\cdot (p-1), \quad s_{i_3} = 3\cdot (p-2), 
    \quad\ldots\quad, \qquad s_{i_{p-1}} = (p-1)\cdot 2,
    \quad  s_{i_{p}} = p .
\]
The mid-point of $(\sss\;|\;K_{p+1})$ is $x_{i_{(p/2)+1}}$.
The two adjacent cutwidths of the mid-point $x_{i_{(p/2)+1}}$ are:
\[
    s_{i_{(p/2)}} = \dfrac{p}{2}\cdot (\dfrac{p}{2}+1)\quad\text{and}\quad
    s_{i_{(p/2)+1}} = (\dfrac{p}{2}+1)\cdot \dfrac{p}{2},
\]
so that also, as one can readily check:
\[
   s_{i_{(p/2)}}\ =\ s_{i_{(p/2)+1}}\ =\ \dfrac{p^2+2p}{4}\ =
   \ \max\,\Set{s_{i_1}, s_{i_2}, \ldots, s_{i_{p}}} ,
\]
and the sequence of cutwidths $(s_{i_1},\ldots,s_{i_p})$ is equal to
its own reverse $(s_{i_p},\ldots,s_{i_1})$.  Moreover, for every $j$
such that $1\leqslant j < i_1$ or $i_{p+1}< j\leqslant n+p$, we have
$s_j = 0$. Also, it is intuitively useful for the argument in the
proof of Lemma~\ref{lem:auxiliary-vertices-together-1} to keep in mind
that:
\begin{alignat*}{10}
 & (s_{i_1} - s_{j})\ && =\quad && p, \qquad 
       &&\text{for every\quad $1\leqslant j < i_1$},
\\
 & (s_{i_2} - s_{j})\ && =\ && p-2, \qquad 
       &&\text{for every\quad $i_1\leqslant j < i_2$},
\\
 & (s_{i_3} - s_{j})\ && =\ && p-4, \qquad 
       &&\text{for every\quad $i_2\leqslant j < i_3$},
\\ &\quad \cdots && && \cdots && \quad \cdots
\\
 & (s_{i_{(p/2)}} - s_{j})\ && =\ && 2, \qquad 
       &&\text{for every\quad $i_{(p/2)-1} \leqslant j < i_{p/2}$},
\\
 & (s_{i_{(p/2)+1}} - s_{j})\ && =\ && 0, \qquad 
       &&\text{for every\quad $i_{(p/2)} \leqslant j < i_{(p/2)+1}$}.
\end{alignat*}

\paragraph{Proof of 
Lemma~\ref{lem:auxiliary-vertices-together-1}.}
\label{proof:auxiliary-vertices-together-1} 
In the sequence $\sss$ in~\eqref{eq:sequence-of-vertices-and-cutwidths}
in Section~\ref{sect:beta-optimization}, suppose:
\begin{itemize}[itemsep=0pt,parsep=5pt,topsep=5pt,partopsep=0pt] 
\item $x_i$ is the leftmost vertex in $U\cup\Set{w}$,
\item $x_j$ is the leftmost vertex in $V-\Set{w}$ to the right of $x_i$,
\item $x_{\ell}$ is the rightmost vertex in $U\cup\Set{w}$,
\item $x_k$ is the rightmost vertex in $V-\Set{w}$ to the left of $x_{\ell}$,
\end{itemize}
where $1\leqslant i\leqslant j\leqslant k\leqslant \ell\leqslant p+n$.
Graphically, $\sss$ can be represented by:
\[
  \underbrace{x_1\quad\cdots\quad x_{i-1}}_{\text{all in $V-\Set{w}$}}
  \quad 
  \underbrace{x_i\quad\cdots\quad x_{j-1}}_{\text{all in $U\cup\Set{w}$}}
  \quad \text{\circled{$x_j$}}
  \quad x_{j+1}
  \quad \cdots \quad \quad  x_{k-1} \text{\circled{$x_{k}$}} \quad
  \underbrace{x_{k+1}\quad\cdots\quad x_{\ell}}_{\text{all in $U\cup\Set{w}$}}
  \quad
  \underbrace{x_{\ell+1}\quad\cdots\quad x_{p+n}}_{\text{all in $V-\Set{w}$}}
\]
The circled vertices, $x_j$ and $x_k$, are in $V-\Set{w}$.
If $\sss$ is scattered, then $1\leqslant i<j$ and/or $k<\ell\leqslant n+p$, 
with the possibility
that $j=k$ in which case there is only one vertex in $V-\Set{w}$ inserted
between all the vertices of $U\cup\Set{w}$. We define:
\[
  \scatter{\sss}\ :=\ \min\,\Set{\,j-i,\;\ell-k\,}\ \geqslant 1 ,
\]
when $\sss$ is scattered.
If $\sss$ is not scattered, we set $\scatter{\sss} := 0$, so that
$\sss$ is scattered iff  $\scatter{\sss} \geqslant 1$. Moreover,
with $p$ even and $p+1$ odd, it is always the case that
$\scatter{\sss} \leqslant p/2$, so that if $\sss$ is scattered, then:
\[
  1\ \leqslant\ \scatter{\sss}\ \leqslant\ \dfrac{p}{2}.
\]
To complete the proof, it suffices to show that if $\sss$ is scattered,
then we can define another sequence $\sss'$ from $\sss$ such that: 
\[
   \beta(\sss') < \beta(\sss)
   \quad\text{and}\quad \scatter{\sss'} < \scatter{\sss}.
\]
We obtain $\sss'$ from $\sss$ as follows:
\begin{itemize}[itemsep=0pt,parsep=5pt,topsep=5pt,partopsep=0pt] 
\item if $j-i\leqslant \ell-k$, remove $x_j$ from the $j$-th position
      and insert it between $x_{i-1}$ and $x_{i}$, 
\item if $j-i > \ell-k$, remove $x_k$ from the $k$-th position
      and insert it between $x_{\ell}$ and $x_{\ell+1}$.
\end{itemize}
With no loss of generality, let $j-i\leqslant \ell-k$. 
The portion of $\sss$ under consideration is therefore:
\[
      (r_{i-1},s_{i-1})\quad 
      \underbrace{x_i \quad (r_{i},s_{i})\quad
      \cdots\quad (r_{j-2},s_{j-2})\quad x_{j-1}}_{\text{all in $U\cup\Set{w}$}}
      \quad (r_{j-1},s_{j-1})\quad \text{\circled{$x_j$}}\quad (r_{j},s_{j}) 
\]
and the order of all the vertices in the new $\sss'$ is:
\[
  \underbrace{x_1\quad\cdots\quad x_{i-1}}_{\text{all in $V-\Set{w}$}}
  \quad \text{\circled{$x_j$}} \quad
  \underbrace{x_i\quad\cdots\quad x_{j-1}}_{\text{all in $U\cup\Set{w}$}}
  \quad x_{j+1}
  \quad \cdots \quad  x_{k-1} \quad \text{\circled{$x_{k}$}} \quad
  \underbrace{x_{k+1}\quad\cdots\quad x_{\ell}}_{\text{all in $U\cup\Set{w}$}}
  \quad
  \underbrace{x_{\ell+1}\quad\cdots\quad x_{p+n}}_{\text{all in $V-\Set{w}$}}
\]
Let $x_j = v \in V-\Set{w}$ and partition $d = \degr{G}{v}$ into
$d = d^L+d^R$, where:
\begin{itemize}[itemsep=0pt,parsep=3pt,topsep=5pt,partopsep=0pt] 
\item $d^L$ is the number of vertices 
      in $\Set{x_1,\ldots,x_{j-1}}\cap V$ which are connected to $v$,
\item $d^R$ is the number of vertices 
      in $\Set{x_{j+1},\ldots,x_{n+p}}\cap V$ 
      which are connected to $v$.
\end{itemize}
There are different cases, depending on:
\begin{itemize}[itemsep=0pt,parsep=5pt,topsep=5pt,partopsep=0pt] 
\item the value of $j-i$ between $1$ and $p/2$,
\item the value of $d^R-d^L$ between $-d$ and $+d$,
\item whether the distinguished vertex $w$ is in $\Set{x_i,\ldots,x_{j-1}}$
      or in $\Set{x_{j+1},\ldots,x_{\ell}}$,
\item whether $v$ is connected to $w$ or not.
\end{itemize} 
We consider only one of the cases, which is also a ``worst case'' to 
explain, and leave to the reader all the other cases, which are simple
variations of this ``worst case''. For the ``worst case'' which we choose
to consider, let:
\begin{enumerate}[itemsep=0pt,parsep=5pt,topsep=5pt,partopsep=0pt] 
\item[(1)]  $j = i+p/2$ so that $j-i = p/2$,
\item[(2)]  $d^L = 0$ so that $d^R-d^L = d$,
\item[(3)]  $w$ is in $\Set{x_{j+1},\ldots,x_{\ell}}$,
\item[(4)]  there is an edge $\set{v\,w}$ connecting $v$ and $w$.
\end{enumerate}
With assumptions (1) to (4), as well as after:
\begin{itemize}[itemsep=0pt,parsep=5pt,topsep=5pt,partopsep=0pt] 
\item  substituting $i+(p/2)$ for $j$,
\item  replacing the sequence of cutwidths
       $s_{i-1},\ s_{i},\ \ldots\ ,\ s_{i+(p/2)-1},\ s_{i+(p/2)},\ s_{i+(p/2)+1}$ 
       \\[1ex]
       by their actual values
       $0,\ p,\ \ldots\ ,\ (p^2 + 2p -8)/4,\ (p^2+2p)/4,\ (p^2+2p)/4$,
       respectively,
\item  and posing $r := r_{i-1}$,
\end{itemize}
the portion of $\sss$ under consideration becomes:
\[
      (r,0)\quad 
      \underbrace{x_i \quad (r,p)\quad
      \cdots\quad \Bigl(r,(p^2 + 2p -8)/4\Bigr)
      \quad x_{i+(p/2)}}_{\text{all in $U\cup\Set{w}$}}
      \quad \Bigl(r,(p^2+2p)/4\Bigr)\quad \text{\circled{$v$}}
      \quad \Bigl(r+d,(p^2+2p)/4\Bigr) 
\]
The corresponding portion in the new $\sss'$ is:
\[
      (r,0)\quad \text{\circled{$v$}} \quad (r+d,0) \quad
      \underbrace{x_i \quad (r+d,p)\quad
      \cdots\quad \Bigl(r+d,(p^2 + 2p -8)/4\Bigr)
      \quad x_{i+(p/2)}}_{\text{all in $U\cup\Set{w}$}}
      \quad \Bigl(r+d,(p^2+2p)/4\Bigr) 
\]
with all the cutwidths to the left and to the right of the shown
portion being identical in $\sss$ and $\sss'$.
It is now readily seen that the value of $\beta(\sss')$ is: 
\[
   \beta(\sss')\ =\ \beta(\sss) - (p^2+2p)/4 + (p/2)\,d
   \ =\ \beta(\sss) - \dfrac{p^2 - 2 (d - 1) p}{4}
\]
Let $\Delta := \sum \Set{\,\degr{G}{x}\,|\,x\in V}$. By the construction
of the auxiliary graph $G_w$, we have $p=\Delta$. The value of 
$d=\degr{G}{v} \leqslant \Delta/2$, the upper bound $\Delta/2$
being the extreme case when $G$ is a star graph with: $v$ at its center,
$\Delta/2$ leaf vertices among $\Set{x_{j+1},\ldots,x_{n+p}}\cap V$,
and all other vertices of $V$ being isolated. 
Hence, 
\[
 p^2 - 2\,(d-1)\,p\ \leqslant\ {\Delta}^2 - 2 (\dfrac{\Delta}{2} - 1) {\Delta}
 \ =\ {\Delta}^2 - {\Delta}^2 + 2\, {\Delta}\ =\ 2\, {\Delta} .
\]
Hence, $\beta(\sss') \leqslant \beta(\sss) - \Delta/2$, so that
$\beta(\sss') < \beta(\sss)$ which is the desired conclusion.
\hfill\QED

\bigskip
For precision in the next proof, we introduce the measure
of \emph{unbalance}.

\begin{definition}{Unbalance}
\label{def:unbalance-factor}
Consider the sequence $\sss$ in~\eqref{eq:sequence-of-vertices-and-cutwidths} in 
Section~\ref{sect:beta-optimization}.
\begin{itemize}
\item Let $a^L\geqslant 0$ be the number of 
      vertices from $V-\Set{w}$ to the left of $w$,
       and $a^R\geqslant 0$ be the number of 
      vertices from $V-\Set{w}$ to the right of $w$.
\item Let $b^L\geqslant 0$ be the number of 
      vertices from $U$ to the left of $w$,
      and $b^R\geqslant 0$ be the number of 
      vertices from $U$ to the right of $w$.
\end{itemize}
The unbalance of $\sss$ is measured by:
\[
   \unbal{\sss}\ :=
   \ \min\,\Set{\,(n-a^L-1) + (p-b^R),\ (n-a^R-1) + (p-b^L)\,} .
\]
The quantity $(n-a^L-1) + (p-b^R)$ measures $\sss$'s unbalance
on the left, and similarly $(n-a^R-1) + (p-b^L)$ measures
$\sss$'s unbalance on the right. It is useful to keep in mind that:
\[
  (p-b^R)+(p-b^L) = p\quad\text{and}\quad
  (n-a^L-1)+(n-a^R-1) = n - 1,
\]
so that, if the quantity $(n-a^L-1) + (p-b^R)$ or the quantity
$(n-a^R-1) + (p-b^L)$ is reduced to $0$, then the other of these
two quantities is increased to $n-1+p$.
\end{definition}

\paragraph{Proof of 
Lemma~\ref{lem:auxiliary-vertices-together-2}.}
\label{proof:auxiliary-vertices-together-2} 
By Lemma~\ref{lem:auxiliary-vertices-together-1}, we can assume
that $\sss$ is not scattered. 
It suffices to show that if $\unbal{\sss} \geqslant 1$, we can 
define another sequence $\sss'$ from $\sss$ such that
$\beta(\sss') < \beta(\sss)$ and $\unbal{\sss'} < \unbal{\sss}$. 
We use the notation in the proof of 
Lemma~\ref{lem:auxiliary-vertices-together-1}. The portion of $\sss$
that we examine closely is:
\[
      (r_{i-1},s_{i-1})\quad 
      \underbrace{x_i \quad (r_{i},s_{i})\quad
      x_{i+1}\quad (r_{i+1},s_{i+1})\quad \cdots\quad
      (r_{i+p-1},s_{i+p-1})\quad
      x_{i+p}}_{\text{all in $U\cup\Set{w}$}}\quad (r_{i+p},s_{i+p})\quad 
\]
where:
\begin{alignat*}{5}
  & V - \Set{w}\ &&=\ && \Set{x_1,\ldots,x_{i-1}}\cup\Set{x_{i+p+1},\ldots,x_{n+p}},
  \quad && \text{with $1\leqslant i\leqslant n$},
\\
  & U\cup\Set{w}\ &&=\ &&  \Set{x_{i},\ldots,x_{i+p}},
   \quad && \text{with $w = x_{i+k}$ and $0\leqslant k\leqslant p$}.
\end{alignat*}
We partition $d = \degr{G}{x_{i+k}} = \degr{G}{w}$ into $d = d^L+d^R$, where:
\begin{itemize}[itemsep=0pt,parsep=3pt,topsep=5pt,partopsep=0pt] 
\item $d^L\geqslant 0$ is the number of vertices 
      in $\Set{x_1,\ldots,x_{i-1}}$ 
      which are connected to $w = x_{i+k}$ in $G$,
\item $d^R\geqslant 0$ is the number of vertices 
      in $\Set{x_{i+p+1},\ldots,x_{n+p}}$ 
      which are connected to $w = x_{i+k}$ in $G$.
\end{itemize}
And we partition $e = \degr{K_{p+1}}{x_{i+k}} = \degr{K_{p+1}}{w}$ into 
$e = e^L+e^R = p$, where:
\begin{itemize}[itemsep=0pt,parsep=3pt,topsep=5pt,partopsep=0pt] 
\item $e^L = k$ is the number of vertices 
      in $\Set{x_{i},\ldots,x_{i+k-1}}$ 
      which are connected to $w = x_{i+k}$ in $K_{p+1}$,
\item $e^R = (p-k)$ is the number of vertices 
      in $\Set{x_{i+k+1},\ldots,x_{i+p}}$ 
      which are connected to $w = x_{i+k}$ in $K_{p+1}$.
\end{itemize}
Though not explicitly used below, it is worth noting
that $e^L$ and $e^R$ here are the same as $b^L$ and $b^R$ in
Definition~\ref{def:unbalance-factor} because $K_{p+1}$ is a complete
graph (but $d^L$ and $d^R$ are not the same as $a^L$ and $a^R$).  
We consider $5$ separate cases, $\Set{(a), (b), (c), (d), (e)}$:
\begin{itemize}[itemsep=0pt,parsep=3pt,topsep=5pt,partopsep=0pt] 
\item[(a)] $k=0$, which implies $e^L = 0$ and $e^R = p$. 
\end{itemize}
In case (a), because $\unbal{\sss}\neq 0$ by hypothesis, it must be
that $\Set{x_{i+p+1},\ldots,x_{n+p}}\neq\varnothing$.  
It suffices to move the vertices in $\Set{x_{i+p+1},\ldots,x_{n+p}}$ 
to the left of $w=x_{i}$, also preserving their order
\[
       x_1,\ \ldots\ ,\ x_{i-1},\ x_{i+p+1},\ \ldots\ , x_{n+p} .
\]
Using a reasoning similar
to that in the proof of Lemma~\ref{lem:auxiliary-vertices-together-1},
we leave it to the reader to show that $\unbal{\sss'} = 0$ and
$\beta(\sss') < \beta(\sss)$ for the resulting sequence $\sss'$.
\begin{itemize}[itemsep=0pt,parsep=3pt,topsep=5pt,partopsep=0pt] 
\item[(b)] $k=p$, which implies $e^L = p$ and $e^R = 0$. 
\end{itemize}
Case (b) is similar to case (a).  Because $\unbal{\sss}\neq 0$ by
hypothesis, it must be that $\Set{x_{1},\ldots,x_{i-1}}\neq\varnothing$.
In this case, we move the vertices in $\Set{x_{1},\ldots,x_{i-1}}$ to
the right of $w=x_{i+p}$.  Again, we leave it to the reader to show
that $\unbal{\sss'} = 0$ and $\beta(\sss') < \beta(\sss)$ for
the resulting sequence $\sss'$.

For the three remaining cases, we can assume that neither $k=0$ nor
$k=p$, \ie, both $w \neq x_{i}$ and $w \neq x_{i+p}$. Two cases of 
these three are:
\begin{itemize}[itemsep=0pt,parsep=3pt,topsep=5pt,partopsep=0pt] 
\item[(c)] $d^L > d^R$, in which case we tranpose $w = x_{i+k}$
           and $x_{i}$. 
\item[(d)] $d^L < d^R$, in which case we transpose $w = x_{i+k}$ 
           and $x_{i+p}$.
\end{itemize}
By a reasoning similar to that in the proof of
Lemma~\ref{lem:auxiliary-vertices-together-1}, we leave to the reader
the straightforward details showing that $\beta(\sss') < \beta(\sss)$
in both case (c) and case (d).
\begin{itemize}[itemsep=0pt,parsep=3pt,topsep=5pt,partopsep=0pt] 
\item[(e)] $d^L = d^R$, in which case the value
           of $\beta(\sss)$ remains unchanged by
           tranposing $w = x_{i+k}$ and $x_{i}$, or by transposing
           $w = x_{i+k}$ and $x_{i+p}$, and so we need an additional
           argument.
\end{itemize}
The additional argument for case (e), is to first transpose
$w = x_{i+k}$ and $x_{i}$, or alternatively 
transpose $w = x_{i+k}$ and $x_{i+p}$, thus reducing case (e) to case (a),
or alternatively reducing case (e) to case (b).
\hfill \QED

\newpage

\Hide
{\footnotesize
\printbibliography
}

{\footnotesize 
\bibliographystyle{plainurl} 
\bibliography{generic,extra}
}

\ifTR
\else
\fi

\end{document}